\documentclass{article}
\usepackage{setspace}
\usepackage{graphicx}
\usepackage{amsmath,amssymb,amsthm}
\usepackage{color}
\usepackage[margin=1in]{geometry}
\usepackage[square,numbers]{natbib}
\usepackage{hyperref}
\usepackage{enumitem}
\usepackage{subcaption}

\newtheorem{theorem}{Theorem}
\newtheorem{lemma}[theorem]{Lemma}
\newtheorem{proposition}[theorem]{Proposition}
\newtheorem{corollary}[theorem]{Corollary}
\newtheorem{assumption}[theorem]{Assumption}
\newtheorem{definition}[theorem]{Definition}
\newtheorem{example}[theorem]{Example}
\newtheorem{remark}[theorem]{Remark}
\numberwithin{equation}{section}
\numberwithin{theorem}{section}

\newcommand{\R}{\mathbb{R}}
\newcommand{\dpi}{\Delta\pi}
\newcommand{\rcal}{\mathcal{R}}
\newcommand{\ccal}{\mathcal{C}}
\newcommand{\lcal}{\mathcal{L}}
\newcommand{\V}{\mathcal{V}}

\title{Oracle-Parametrized Constant Function Market Makers:\\ From Price Feeds to Pricing Rules}
\author{Hamed Amini \thanks{Center for Applied Optimization, Department of Industrial and Systems Engineering, University of Florida, Gainesville, FL, USA. {\tt aminil@ufl.edu}}
\and
Zachary Feinstein \thanks{Stevens Institute of Technology, School of Business, Hoboken, NJ 07030, USA. {\tt zfeinste@stevens.edu}. 
The author acknowledges the support from NSF IUCRC CRAFT center research grant (2113906) for this research. The opinions expressed in this publication do not necessarily represent the views of NSF IUCRC CRAFT.}}
\date{\today}

\begin{document}

\maketitle

\begin{abstract}
This paper introduces oracle-parametrized automated market makers (OP-AMMs), i.e., automated market makers whose quoted price depends jointly on the pool reserves and an external oracle price. In doing so, we extend the information-agnostic AMM framework to settings, such as tokenized securities, for which price discovery occurs off-chain. Under a strict oracle-contraction condition, we show that the quoted price of any OP-AMM interpolates between the oracle price and an implicit autarkic price determined by the pool reserves. We then derive a general loss-versus-rebalancing (LVR) decomposition that separates the residual exposure to market lags from the losses induced by oracle errors. This analysis is further extended to stale, discrete-update oracles and to sandwich attacks around oracle updates. Using this framework, we find conditions under which OP-AMMs simultaneously increase local capital efficiency and reduce normalized LVR relative to information-agnostic AMMs. However, sufficiently noisy or stale oracles can reverse these gains. A counterfactual backtest using one-second SPY NBBO data is provided to demonstrate these trade-offs. In particular, we map the Pareto-efficient frontier of oracle-parametrized constant function market maker (OP-CFMM) designs across stylized oracle regimes.
\end{abstract}

\section{Introduction}\label{sec:intro}
Decentralized finance (DeFi) is the field of using blockchain technology to provide financial services through open, rule-based protocols; in doing so, DeFi democratizes access to financial intermediation. Though DeFi was initially applied primarily to cryptocurrencies and stablecoins, the tokenization of real-world assets (RWAs) has begun to accelerate over the past year amid increasing regulatory clarity, with the scale of tokenized assets projected to reach trillions of dollars by 2030~\cite{mckinsey,carapella2023tokenization,watsky2024tokenized}. Combining RWAs with DeFi protocols would extend this intermediation to traditional assets. However, the prices of tokenized RWAs are discovered off-chain, in the markets for the underlying assets; this separation can create extractable value within DeFi.

The dominant architecture for decentralized exchanges (DEXes) is that of automated market makers (AMMs). Briefly, these mechanisms quote marginal prices to any user based solely on the inventory of the AMM pool~\cite{angeris2021constant,angeris2023geometry}; the realized cost of a transaction is then obtained by integrating these marginal prices over the corresponding change in pool reserves, thereby endogenizing price impacts (see, e.g.,~\cite{bichuch2022axioms,schlegel2022axioms,lee2023all,frongillo2023axiomatic} for axiomatic treatments of AMMs). In this sense, conventional AMMs can be described as \emph{information agnostic} since they observe only their own token reserves and incorporate no external information into their pricing rules. This price autarky causes AMMs to act as passive market makers and leaves them exposed to stale-price arbitrage when price discovery occurs externally, as is the case for RWAs. The resulting arbitrage losses are quantified by metrics such as loss-versus-rebalancing~\cite{milionis2022automated,cartea2023predictable}.

In contrast to conventional AMMs, proprietary AMMs (Prop AMMs) have recently emerged that update their quoted prices at high frequency so as to provide deep liquidity while mitigating their exposure to adverse selection.\footnote{\url{https://dune.com/dflow/prop-amms}} As the name suggests, these systems rely on proprietary pricing mechanisms and place market making operations in the hands of specialized firms. As such, and in contrast to the democratizing ethos of DeFi, Prop AMMs represent a reprofessionalization of on-chain market making.

Within this work, we combine the public liquidity provision of AMMs with the informed pricing of Prop AMMs by studying AMMs parametrized by external pricing oracles. Bespoke oracle-dependent AMMs have been proposed and developed previously, including the proactive market maker (PMM) design of DODO~\cite{dodo,chen2023improving}, Curve v2~\cite{curvev2}, UAMM~\cite{im2023uamm}, and the dynamic curves of~\cite{krishnamachari2021dynamic}.\footnote{The dynamic-curve construction of~\cite{krishnamachari2021dynamic} re-anchors the trading curve at the current reserve state so that its marginal price matches the oracle, whereas the framework herein specifies a single AMM pricing rule conditional on each oracle value.} Further, our framework is complementary to that of~\cite{bergault2024automated,bergault2024price}, which take an optimization-based perspective on price-aware market making. Herein, we propose an axiomatic, mechanism-level framework that characterizes decentralized oracle-parametrized pricing rules.

Near-oracle pricing is suggested within \cite[Section 10]{milionis2022automated} as a way to reduce loss-versus-rebalancing. That work observes that an AMM with access to a high-frequency oracle could quote prices arbitrarily close to the external price and thereby approach the payoff of the rebalancing benchmark. However, it also cautions that such a design relies heavily on the accuracy of the oracle and leaves open the potential for its manipulation. That proposal is closest to a quote that moves one-for-one with the oracle and, therefore, transmits the entire oracle error to the liquidity providers. Notably, in formalizing these designs, we do \emph{not} assume a perfect oracle. Instead, we quantify the losses associated with oracle errors and sandwich attacks around oracle updates (see, e.g.,~\cite{eskandari2021sok}). In doing so, we find a trade-off between liquidity efficiency and the risks created by the dependence on an oracle.

We wish to highlight that the risk of oracle manipulation may be less acute for tokenized securities than for cryptocurrencies. The vulnerabilities documented for crypto oracles are largely concentrated in long-tail, low-liquidity assets for which the spot price (and the time-weighted average price) is inexpensive to move. In contrast, the reference price for a tokenized security is discovered on deep off-chain markets subject to regulatory surveillance, e.g., the national best bid and offer (NBBO) of an equity; in such markets, distorting the underlying quote can be prohibitively costly and may entail legal exposure. Further, modern pull-oracle architectures, such as those offered by Chainlink and Pyth, reduce the opportunities for manipulation by quoting recent off-chain data directly, thereby removing the stale on-chain values that can be, e.g., front-run.

The primary contributions of this work are threefold. First, in Section~\ref{sec:pmm}, we \textbf{construct a general class of oracle-parametrized AMMs} (OP-AMMs) that satisfy the desirable properties of AMMs presented in, e.g.,~\cite{bichuch2022axioms,schlegel2022axioms}. Under a strict oracle-contraction condition, we prove that any such OP-AMM can be written as an interpolation between the oracle price and an implicit autarkic price. Second, in Section~\ref{sec:lvr}, we quantify the \textbf{loss-versus-rebalancing (LVR) of OP-AMMs} and its dependence on oracle risk. Specifically, we decompose the residual price variation underlying LVR into market, oracle, and covariance components. We further extend this analysis to discrete oracle updates, identify the resulting oracle-refresh losses as oracle extractable value, and study sandwich attacks around these updates along with the protection provided by transaction fees. Finally, in Section~\ref{sec:cases}, we investigate the \textbf{capital efficiency and arbitrage incentives} of OP-AMMs through an analytical comparison and a counterfactual backtest using historical SPY market data. These results suggest that, with a sufficiently reliable oracle, OP-AMMs can reduce stale-price arbitrage (and the associated LVR) while simultaneously deepening the local liquidity available to liquidity-motivated order flow; noisy or stale oracles, however, can reverse these gains. Section~\ref{sec:discussion} discusses the economic implications of OP-AMMs, and Section~\ref{sec:conclusion} concludes. The proofs of all results are provided within the appendix.

\section{Oracle-Parametrized Automated Market Makers}\label{sec:pmm}

Within this section, we introduce an axiomatic framework for oracle-parametrized automated market makers (OP-AMMs). This framework extends the axiomatic theory of AMMs developed in~\cite{bichuch2022axioms} by allowing the pricing rule to depend explicitly on an exogenous price signal, i.e., an oracle price. Throughout this work, we focus on markets with two assets and impose structural conditions to guarantee economically meaningful pricing behavior. In Section~\ref{sec:pmm-construction}, we define OP-AMMs, derive their conditional invariants, and present constructions built from a baseline AMM. The characterization of OP-AMMs as an interpolation between the oracle price and an internal equilibrium price is provided in Section~\ref{sec:pmm-characterization}. Finally, Section~\ref{sec:pmm-liquidity} measures the local liquidity of an OP-AMM through an equivalent constant-product market maker (CPMM) liquidity parameter.

\subsection{Construction}\label{sec:pmm-construction}

We begin with some simple notation that will be used throughout the paper. Except where otherwise stated, we will work with log-prices. Let $\pi\in\R$ denote the oracle \emph{log}-price. The internal state of the market maker is summarized by the reserve ratio $r:=\frac{y}{x}\in\R_{++}$, which reflects the scale invariance of the reserve holdings $(x,y)\in\R_{++}^2$. For a function $p$ of the reserve ratio and oracle price, we denote its partial derivatives by subscripts, e.g., $p_r := \partial_r p$ and $p_\pi := \partial_\pi p$. Rather than beginning from an invariant, an OP-AMM is defined directly through its quoted log-price.

\begin{definition}\label{defn:PMM}
An \textbf{\emph{oracle-parametrized automated market maker (OP-AMM)}} is a continuously differentiable function $p:\R_{++}\times\R\to\R$, where $p(r,\pi)$ denotes the quoted log-price at reserve ratio $r$ and oracle log-price $\pi$, satisfying the following conditions:
\begin{enumerate}[label=(\roman*)]
    \item \textbf{Strict monotonicity in the reserve ratio:} For each fixed $\pi\in\R$, the function $p(\cdot,\pi)$ is strictly increasing.

    \item \textbf{Surjectivity in the reserve ratio:} For each fixed $\pi\in\R$,
    $$\lim_{r\searrow0}p(r,\pi)=-\infty
    \qquad \text{ and } \qquad
    \lim_{r\nearrow\infty}p(r,\pi)=\infty.$$

    \item \textbf{Non-decreasing and non-expansive oracle response:} For each fixed $r>0$, the function $p(r,\cdot)$ is non-decreasing and non-expansive, i.e., for all $\pi_1\leq\pi_2$,
    $$0 \leq p(r,\pi_2)-p(r,\pi_1) \leq \pi_2-\pi_1.$$
    As $p$ is continuously differentiable, this is equivalent to $0\leq p_\pi(r,\pi)\leq1$ for every $\pi \in \R$.
\end{enumerate}
\end{definition}

\begin{remark}\label{rem:PMM-conditions}
Strict monotonicity and surjectivity in the reserve ratio imply that, for each fixed oracle price, every finite quoted log-price corresponds to a unique reserve ratio. Monotonicity in the oracle price guarantees that the OP-AMM responds in the same direction as the oracle signal, while non-expansiveness prevents the pricing rule from amplifying changes in that signal.
\end{remark}

In contrast to classical AMMs, an OP-AMM generally does not admit a single oracle-independent invariant since the oracle price may change exogenously even in the absence of trades. However, when the oracle price is held fixed, the OP-AMM behaves as a standard AMM, i.e., trading is path independent and preserves an invariant conditional on that oracle price. To express this conditional invariant, it will be convenient to consider the integral mapping
\begin{equation}\label{eq:G}
G(r,\pi) := \int_1^r \frac{dq}{q + \exp(p(q,\pi))}
\end{equation}
for any $(r,\pi) \in \R_{++} \times \R$, together with its two boundary limits
\begin{equation}\label{eq:g-limits}
g_0(\pi) := \lim_{r \searrow 0} G(r,\pi)
\qquad \text{ and } \qquad
g_\infty(\pi) := \lim_{r \nearrow \infty} \bigl[G(r,\pi) - \log r\bigr].
\end{equation}
Importantly, both limits are well-defined in $[-\infty,0)$ as $r \in \R_{++} \mapsto G(r,\pi)$ is strictly increasing and $r \in \R_{++} \mapsto G(r,\pi) - \log r$ is strictly decreasing by inspection. As Definition~\ref{defn:PMM} specifies only a quoted price, the following proposition derives the structure of this conditional invariant within the geometric framework of~\cite[Section 1]{angeris2023geometry}.

\begin{proposition}\label{prop:CI}
Let $p: \R_{++} \times \R \to \R$ be an OP-AMM. Define $U: \R^2_+ \times \R \to \R \cup \{-\infty\}$ as the continuous extension of
$$U(x,y,\pi) := \log x + G(y/x,\pi)$$
for any $(x,y,\pi) \in \R^2_{++} \times \R$. The zero-superlevel set
$$\rcal(\pi) := \bigl\{(x,y) \in \R^2_+ \mid U(x,y,\pi) \geq 0\bigr\}$$
is a nondegenerate reachable set for any $\pi \in \R$, i.e., $\emptyset \neq \rcal(\pi) \subseteq \R^2_+ \setminus \{(0,0)\}$ is closed, convex, and upward closed. Moreover, for fixed $\pi \in \R$, the exponentiated invariant $\lcal(x,y,\pi) := \exp U(x,y,\pi)$, defined for $(x,y,\pi) \in \R^2_+ \times \R$, is the canonical trading function of $\rcal(\pi)$, i.e., the unique non-decreasing, concave, and positively homogeneous function whose unit-superlevel set is $\rcal(\pi)$. Finally, the marginal price of this derived invariant corresponds to the OP-AMM, i.e., $\partial_x U(x,y,\pi) / \partial_y U(x,y,\pi) = \exp(p(y/x,\pi))$ for any $(x,y) \in \R^2_{++}$.
\end{proposition}
\begin{proof}
See Appendix~\ref{app:proof-CI}.
\end{proof}

In addition to the conditions of Definition~\ref{defn:PMM}, it is often natural to impose further structural properties on an OP-AMM.

\begin{enumerate}[label=(\roman*),start=4]

\item \textbf{Surjectivity in the oracle price:} For each fixed $r>0$,
$$\lim_{\pi\searrow-\infty}p(r,\pi)=-\infty
\qquad \text{ and } \qquad
\lim_{\pi\nearrow\infty}p(r,\pi)=\infty.$$
This optional condition excludes information-agnostic AMMs.

\item \textbf{Boundary divergence:} For each fixed $\pi\in\R$,
$$g_0(\pi) = g_\infty(\pi) = -\infty.$$
By the conditional-invariant representation of Proposition~\ref{prop:CI}, this condition ensures that approaching either boundary $r\searrow0$ or $r\nearrow\infty$ along a conditional-invariant curve requires infinite cumulative trade volume.

\item \textbf{Symmetry:} The pricing rule satisfies
$$p(r,\pi) = -p(1/r,-\pi),$$
which reflects invariance under relabeling of the two assets.

\item \textbf{Re-denomination invariance:} For any $a>0$, the pricing rule satisfies
$$p(r,\pi) = p(ar,\pi+\log a)-\log a,$$
so that pricing is consistent under a change of units.
\end{enumerate}

By Proposition~\ref{prop:CI}, every OP-AMM admits an oracle-indexed family of conditional invariants $U: \R^2_+ \times \R \to \R \cup \{-\infty\}$. When these conditional invariants are all generated by a single baseline AMM pricing function, we refer to the OP-AMM and its conditional invariants jointly as an \emph{oracle-parametrized constant function market maker} (OP-CFMM). That is, an OP-CFMM is an OP-AMM which is built from a baseline AMM design.

We now present two such constructions. Let $u:\R_{++}^2\to\R$ be the log of a differentiable, strictly increasing canonical trading function~\cite[Section 1.3]{angeris2023geometry}. As a canonical trading function is positively homogeneous, the marginal price depends only on the reserve ratio $r:=y/x$, and we define the associated AMM pricing function by
$$\phi(r) := \frac{\partial_xu(1,r)}{\partial_yu(1,r)}, \qquad r\in\R_{++}.$$
For the remainder of this paper, we will assume that $\phi:\R_{++}\to\R_{++}$ is continuously differentiable, strictly increasing, and surjective, and we normalize it so that $\phi(1)=1$.

\begin{example}[Price-Tracking OP-CFMM (PT)]\label{ex:PTPMM}
Consider the pricing rule
$$p_{\mathrm{PT}}(r,\pi) := \pi+\log\phi(r), \qquad (r,\pi)\in\R_{++}\times\R.$$
Then $p_{\mathrm{PT}}$ defines an OP-AMM. Indeed, for each fixed $\pi\in\R$, the map $r\mapsto p_{\mathrm{PT}}(r,\pi)$ is strictly increasing and surjective because $\phi$ is strictly increasing and surjective. For each fixed $r>0$, the map $\pi\mapsto p_{\mathrm{PT}}(r,\pi)$ is strictly increasing, non-expansive, and surjective with $\partial_\pi p_{\mathrm{PT}}(r,\pi)=1$. We note that symmetry holds if and only if the baseline satisfies $\phi(r)\phi(1/r)=1$, e.g., for the CPMM $\phi(r)=r$, whereas re-denomination invariance fails for every baseline AMM.

At balanced reserves $r=1$, the normalization $\phi(1)=1$ implies $p_{\mathrm{PT}}(1,\pi)=\pi$, i.e., the quoted price coincides with the oracle price. For fixed $\pi$, the associated conditional invariant is
$$U_{\mathrm{PT}}(x,y,\pi) = \log x + \int_1^{y/x} \frac{dq}{q+e^\pi\phi(q)}.$$

If $\phi(r)=r$, i.e., the baseline is the CPMM $u(x,y)=\frac{1}{2}\log x+\frac{1}{2}\log y$, then $p_{\mathrm{PT}}(r,\pi) = \pi+\log r$ and
$$U_{\mathrm{PT}}(x,y,\pi) = \log x + \int_1^{y/x} \frac{dq}{q+e^\pi q} = \frac{e^\pi}{1+e^\pi}\log x + \frac{1}{1+e^\pi}\log y.$$
Thus, the level sets of $U_{\mathrm{PT}}$ coincide with those of the Balancer weighted-product invariant $x^{w_x(\pi)}y^{w_y(\pi)} = \mathrm{constant}$, with oracle-dependent weights $w_x(\pi) = \frac{e^\pi}{1+e^\pi}$ and $w_y(\pi) = \frac{1}{1+e^\pi}$. When $\pi=0$, the weights are equal and the level sets reduce to those of the Uniswap~V2 CPMM.
\end{example}

Motivated by the discussion in \cite[Remark~14]{bichuch2022axioms}, we next consider a construction in which the inventory reference point of the AMM shifts with the external oracle price.

\begin{example}[Inventory-Tracking OP-CFMM (IT)]\label{ex:ITPMM}
Consider the pricing rule
$$p_{\mathrm{IT}}(r,\pi) := \pi+\log\phi(re^{-\pi}), \qquad (r,\pi)\in\R_{++}\times\R.$$
We will show that $p_{\mathrm{IT}}$ satisfies the defining OP-AMM conditions under a bound on the elasticity of $\phi$. For each fixed $\pi\in\R$, strict monotonicity and surjectivity in the reserve ratio follow from the corresponding properties of $\phi$. For each fixed $r>0$, set $q:=re^{-\pi}$. Then $\partial_\pi p_{\mathrm{IT}}(r,\pi) = 1-q\phi'(q)/\phi(q)$. Consequently, the oracle response is non-decreasing and non-expansive provided $0 \leq q\phi'(q)/\phi(q) \leq 1$ for $q>0$.
Furthermore, the optional oracle-surjectivity property holds under the additional tail conditions
$$\lim_{q\searrow0}\frac{\phi(q)}{q} = \infty, \qquad \lim_{q\nearrow\infty}\frac{\phi(q)}{q} = 0.$$
We note that, as with PT, symmetry is equivalent to $\phi(q)\phi(1/q)=1$; re-denomination invariance, however, now holds for every baseline AMM.

For a fixed oracle price $\pi$, the associated conditional invariant is
$$U_{\mathrm{IT}}(x,y,\pi) := u(e^\pi x,y) - u(e^\pi,1).$$
Indeed, its marginal price is
$$\frac{\partial_xU_{\mathrm{IT}}(x,y,\pi)}{\partial_yU_{\mathrm{IT}}(x,y,\pi)} = e^\pi \phi\left(\frac{y}{e^\pi x}\right) = e^\pi\phi(re^{-\pi}),$$
whose logarithm is $p_{\mathrm{IT}}(r,\pi)$.

If $\phi(r)=r$, i.e., the CPMM, then $p_{\mathrm{IT}}(r,\pi) = \log r$ and the oracle dependence disappears. In this degenerate boundary case, the construction reduces to the original CPMM. As the quoted price is independent of the oracle input, the oracle-surjectivity property does not hold in this case.
\end{example}

\subsection{Characterization Theorem}\label{sec:pmm-characterization}

Within this section, we provide a structural characterization of OP-AMMs that satisfy a strict oracle-contraction condition. Specifically, we will show that the quoted price interpolates between the external oracle price and an internal equilibrium price determined by the inventory of the pool. Recall that, for an OP-AMM $p:\R_{++}\times\R\to\R$, we denote by $p_\pi(r,\pi) := \partial_\pi p(r,\pi)$ the sensitivity of the quoted price to the oracle input.

\begin{assumption}\label{ass:contraction}
For each fixed $r>0$, the map $\pi\mapsto p(r,\pi)$ is a contraction on $\R$. Equivalently, since $p$ is continuously differentiable and non-decreasing in the oracle input, for each $r>0$ there exists a constant $c_r\in[0,1)$ such that
$$0 \leq p_\pi(r,\pi) \leq c_r < 1, \qquad \pi\in\R.$$
\end{assumption}

Intuitively, Assumption~\ref{ass:contraction} requires that the quoted price never fully passes through a change in the oracle price, so that the quote always retains some dependence on the pool reserves. As such, the price-tracking construction of Example~\ref{ex:PTPMM}, for which $p_\pi\equiv1$, is a full pass-through boundary case and is not covered by Assumption~\ref{ass:contraction}. The following proposition guarantees the existence of a unique internal equilibrium price for every reserve ratio under this assumption.

\begin{proposition}\label{prop:fixedpoint}
Under Assumption~\ref{ass:contraction}, for each $r>0$, there exists a unique price $\pi^*(r)\in\R$ such that
$$p(r,\pi^*(r)) = \pi^*(r).$$
\end{proposition}

\begin{proof}
Fix $r>0$. By Assumption~\ref{ass:contraction}, the map $\pi\mapsto p(r,\pi)$ is a contraction on the complete metric space $\R$. The result follows directly from the Banach fixed-point theorem.
\end{proof}

The price $\pi^*(r)$ is the internal equilibrium price implied by the inventory of the OP-AMM. Moreover, by the implicit function theorem, $\frac{d\pi^*(r)}{dr} = \frac{p_r(r,\pi^*(r))}{1-p_\pi(r,\pi^*(r))}$. Thus, wherever $p_r(r,\pi^*(r))>0$, the internal equilibrium price is strictly increasing in the reserve ratio. This brings us to the main result of this section: the quoted price of any strictly contractive OP-AMM is an interpolation between the oracle price and its internal equilibrium price. The proof of this theorem is provided in Appendix~\ref{app:proof-interp}.

\begin{theorem}\label{thm:interp}
Under Assumption~\ref{ass:contraction}, let $\pi^*(r)$ denote the unique fixed point defined in Proposition~\ref{prop:fixedpoint}. For every $r>0$, define
$$\lambda(r,\pi) := \begin{cases}
\displaystyle \frac{p(r,\pi)-\pi^*(r)}{\pi-\pi^*(r)}, & \pi\neq\pi^*(r), \\
p_\pi(r,\pi^*(r)), & \pi=\pi^*(r).
\end{cases}$$
Then $\lambda(r,\pi)\in[0,c_r]$, the map $\pi\mapsto\lambda(r,\pi)$ is continuous, and
$$p(r,\pi) = \lambda(r,\pi)\,\pi + \bigl(1-\lambda(r,\pi)\bigr)\pi^*(r).$$
\end{theorem}

\begin{proof}
See Appendix~\ref{app:proof-interp}.
\end{proof}

Theorem~\ref{thm:interp} shows that every strictly contractive OP-AMM price is a convex combination of the oracle price and an inventory-based equilibrium price. The coefficient $\lambda(r,\pi)$ measures the fraction of the oracle displacement that is transmitted to the quoted price. We next provide a parametric family of OP-CFMMs that interpolates between an inventory-based AMM and a fully oracle-priced design.

\begin{example}[Geometric-Average OP-CFMM (GA)]\label{ex:GAPMM}
Let $\phi:\R_{++}\to\R_{++}$ be the baseline AMM pricing function introduced above, i.e., strictly increasing and surjective with the normalization $\phi(1)=1$. Fix $\lambda\in(0,1)$ and define
$$p_{\mathrm{GA}}(r,\pi) := \lambda\pi + (1-\lambda)\log\phi(r), \qquad (r,\pi)\in\R_{++}\times\R.$$
Then $p_{\mathrm{GA}}$ defines an OP-AMM; the conditions in the reserve ratio follow directly from those imposed on $\phi$, while $\pi\mapsto p_{\mathrm{GA}}(r,\pi)$ is strictly increasing, contractive, and surjective with constant oracle sensitivity $p_\pi(r,\pi) = \lambda$. As with PT and IT, symmetry is equivalent to $\phi(r)\phi(1/r)=1$, while re-denomination invariance holds for GA if and only if the baseline AMM is the CPMM $\phi(r)=r$.

Its internal equilibrium price is $\pi^*(r) = \log\phi(r)$ and the interpolation weight of Theorem~\ref{thm:interp} satisfies $\lambda(r,\pi) \equiv \lambda$, i.e., the quoted price is a fixed convex combination of the oracle price and the AMM-implied equilibrium price. Notably, the limiting cases recover familiar designs: as $\lambda\to0$ the oracle dependence vanishes and the GA converges to the baseline AMM, while as $\lambda\to1$ the pricing rule becomes fully oracle-driven.

When $\phi(r)=r$, i.e., the CPMM, the pricing rule becomes
$$p_{\mathrm{GA}}(r,\pi) = \lambda\pi + (1-\lambda)\log r.$$
In this case, the conditional invariant of Proposition~\ref{prop:CI} is
$$U_{\mathrm{GA}}(x,y,\pi) = \frac{1}{\lambda} \log\left( \frac{y^\lambda+e^{\lambda\pi}x^\lambda}{1+e^{\lambda\pi}} \right).$$
As the denominator depends only on the fixed oracle input, the same trading curves are generated by
$$\widetilde U_{\mathrm{GA}}(x,y,\pi) = \frac{1}{\lambda} \log\bigl( y^\lambda+e^{\lambda\pi}x^\lambda \bigr).$$
This is a logarithmic transformation of a constant elasticity of substitution (CES) utility function; equivalently, the same pricing rule follows from the inventory-tracking construction of Example~\ref{ex:ITPMM} with a CES baseline $q\mapsto q^{1-\lambda}$.
\end{example}

\subsection{Local Liquidity}\label{sec:pmm-liquidity}

Within this section, we compare the local liquidity, i.e., the market depth, of an OP-AMM with that of an information-agnostic AMM. Following \cite[Proposition~3.10]{bichuch2025price}, we quantify local liquidity through the curvature of the trading curve. For a curve represented locally as $y=y(x)$, its curvature is
$$\kappa = \frac{|y''(x)|}{\bigl(1+(y'(x))^2\bigr)^{3/2}}.$$
Lower curvature corresponds to deeper local liquidity and smaller price impact. As such, we will measure the local liquidity of an OP-AMM by matching the curvature of its trading curve to that of a CPMM. The following lemma provides this equivalent CPMM liquidity.

\begin{lemma}\label{lemma:liquidity}
Consider a regular OP-AMM state with quoted price $S>0$, i.e., $p(r,\pi)=\log S$, such that $p_r(r,\pi)>0$, and normalize the pool value so that $Sx+y=1$. Let $L^{\mathrm{OP}}$ denote the CPMM liquidity parameter producing the same local trading-curve curvature as the OP-AMM. Then
$$L^{\mathrm{OP}} = \frac{2\sqrt{S}}{p_r(r,\pi)(S+r)^2} = \frac{4S}{p_r(r,\pi)(S+r)^2}\,L,$$
where $L=\frac{1}{2\sqrt{S}}$ is the liquidity of the CPMM under the same pool-value normalization. Consequently, the OP-AMM is locally more liquid than the equally capitalized CPMM whenever
$$p_r(r,\pi) < \frac{4S}{(S+r)^2}.$$
\end{lemma}

\begin{proof}
See Appendix~\ref{app:proof-liquidity}.
\end{proof}

\begin{example}\label{ex:GA-liquidity}
Consider the GA of Example~\ref{ex:GAPMM}, which admits an explicit liquidity comparison with the baseline AMM. At a state satisfying $p_{\mathrm{GA}}(r,\pi) = \log S$, define $A_\lambda(S,\pi) := S^{\frac{1}{1-\lambda}} e^{-\frac{\lambda}{1-\lambda}\pi}$ and
$$r_\lambda(S,\pi) := \phi^{-1}\bigl(A_\lambda(S,\pi)\bigr), \qquad r_\phi(S) := \phi^{-1}(S).$$
The equilibrium reserve ratio of the GA is therefore $r = r_\lambda(S,\pi)$. As $p_r^{\mathrm{GA}}(r,\pi) = (1-\lambda)\frac{\phi'(r)}{\phi(r)}$, Lemma~\ref{lemma:liquidity} yields
$$L_{\mathrm{GA}} = \frac{2\sqrt{S}\,A_\lambda(S,\pi)}{(1-\lambda) \phi'(r_\lambda(S,\pi)) \bigl(S+r_\lambda(S,\pi)\bigr)^2}.$$
Let $L_\phi$ denote the equivalent CPMM liquidity of the baseline AMM at the same quoted price and pool-value normalization. Then
$$\frac{L_{\mathrm{GA}}}{L_\phi} = \frac{1}{1-\lambda} \frac{A_\lambda(S,\pi)}{S} \frac{\phi'(r_\phi(S))}{\phi'(r_\lambda(S,\pi))} \left( \frac{S+r_\phi(S)}{S+r_\lambda(S,\pi)} \right)^2.$$
If the oracle is accurate, i.e., $\pi = \log S$, then $A_\lambda(S,\pi) = S$ and $r_\lambda(S,\pi) = r_\phi(S)$; therefore, $\frac{L_{\mathrm{GA}}}{L_\phi} = \frac{1}{1-\lambda} > 1$. That is, under a correct oracle signal, the GA is locally more liquid than the equally capitalized baseline AMM.

In the CPMM case $\phi(r)=r$, this comparison becomes explicit. Define $m := \bigl( Se^{-\pi} \bigr)^{\frac{\lambda}{1-\lambda}}$. Then $\frac{L_{\mathrm{GA}}}{L_\phi} = \frac{1}{1-\lambda} \frac{4m}{(1+m)^2}$. Consequently, $L_{\mathrm{GA}}>L_\phi$ if and only if $\frac{1-\sqrt{\lambda}}{1+\sqrt{\lambda}} < m < \frac{1+\sqrt{\lambda}}{1-\sqrt{\lambda}}$ or, equivalently,
$$|\log S-\pi| < \frac{1-\lambda}{\lambda} \log\left( \frac{1+\sqrt{\lambda}}{1-\sqrt{\lambda}} \right).$$
That is, the CPMM-based GA is locally more liquid than the baseline CPMM whenever the oracle price is sufficiently close to the external market price; more accurate oracle signals therefore permit larger values of $\lambda$ and stronger liquidity concentration.
\end{example}

\section{Loss-Versus-Rebalancing}\label{sec:lvr}

Within this section, we study the losses incurred by liquidity providers in an OP-AMM relative to a frictionless rebalancing strategy that holds the same instantaneous risky-asset inventory as the pool but trades at the external market price. This shortfall is known as loss-versus-rebalancing (LVR)~\cite{milionis2022automated}. Whereas the LVR of a standard AMM depends only on the external price process, the LVR of an OP-AMM also depends on the quality and dynamics of the oracle signal. We first derive the LVR of an OP-AMM under a continuously updated oracle in Section~\ref{sec:lvr-general}, which includes both perfect and noisy oracles as special cases. We then consider a piecewise-constant stale oracle with discrete updates in Section~\ref{sec:lvr-stale}; this setting is further used to study sandwich attacks around oracle updates in Section~\ref{sec:lvr-oev}.

\subsection{Continuous Oracle Updates}\label{sec:lvr-general}

Consider the continuous-time performance of an OP-AMM in the absence of trading fees. Let $S_t$ denote the external market price of the risky asset and write $s_t:=\log S_t$. As in~\cite{milionis2022automated}, we will assume that $S_t$ is a continuous local martingale satisfying
\begin{align}\label{eq:st}
\frac{dS_t}{S_t} = \sigma_t\,dW_t, \qquad
ds_t = -\frac{1}{2}\sigma_t^2\,dt + \sigma_t\,dW_t,
\end{align}
where $W_t$ is a standard Brownian motion. Let $\pi_t$ denote the oracle log-price and suppose that
$$d\pi_t = \mu_t^\pi\,dt + \sigma_t^\pi\,dW_t^\pi, \qquad d\langle W,W^\pi\rangle_t = \rho_t\,dt,$$
where the drift $\mu_t^\pi$ may accommodate, e.g., mean reversion or systematic bias, $\sigma_t^\pi$ is the instantaneous oracle volatility, and $\rho_t\in[-1,1]$ denotes the instantaneous correlation between the market and oracle innovations.

At each time $t$, arbitrageurs trade against the OP-AMM until its quoted log-price matches the external log-price, i.e., $p(r_t,\pi_t)=s_t$. As $p(\cdot,\pi_t)$ is strictly increasing and surjective, this clearing condition uniquely determines the post-clearing reserve ratio $r_t=r^*(s_t,\pi_t)$. After observing the current oracle input, arbitrageurs move the reserves along the conditional trading curve associated with $\pi_t$ until this clearing condition is restored. We denote the resulting post-clearing reserves by $(x_t,y_t)$ with $y_t=r_tx_t$. The marked-to-market value of the pool is
$$V_t = x_tS_t+y_t = x_t(S_t+r_t).$$

To define the rebalancing benchmark, let
$$R_t := V_0+\int_0^t x_u\,dS_u.$$
This strategy holds the same instantaneous risky-asset inventory $x_t$ as the pool but executes all reserve adjustments at the external price. We define the cumulative LVR by
$$\mathrm{LVR}_t := R_t-V_t.$$
Finally, define the residual price-innovation process $\zeta_t$ by
$$d\zeta_t := ds_t - p_\pi(r_t,\pi_t)\,d\pi_t, \qquad
\zeta_0:=0.$$
Its instantaneous quadratic-variation rate is
$$\frac{d}{dt}\langle\zeta\rangle_t = \sigma_t^2 - 2\rho_t\,p_\pi(r_t,\pi_t)\sigma_t\sigma_t^\pi + p_\pi(r_t,\pi_t)^2(\sigma_t^\pi)^2.$$
The following theorem provides the LVR of an OP-AMM under this general oracle process. The proof of this theorem is provided within Appendix~\ref{app:lvr-proof}.

\begin{theorem}\label{thm:lvr}
Suppose that $p\in C^2(\R_{++}\times\R)$, $p_r(r,\pi)>0$, and the underlying processes satisfy the usual integrability conditions. Then
$$dV_t = dR_t - d\mathrm{LVR}_t = x_t\,dS_t - \ell_t\,dt,$$
where the instantaneous LVR rate is
\begin{align}\label{eq:LVR}
\ell_t = \frac{1}{2} \frac{S_tx_t}{S_t+r_t} \frac{1}{p_r(r_t,\pi_t)} \frac{d}{dt}\langle\zeta\rangle_t = \frac{S_tx_t}{2(S_t+r_t)p_r(r_t,\pi_t)} \bigl( \sigma_t^2 - 2\rho_t\,p_\pi(r_t,\pi_t)\sigma_t\sigma_t^\pi + p_\pi(r_t,\pi_t)^2(\sigma_t^\pi)^2 \bigr).
\end{align}
In particular, $\ell_t\geq0$ and the cumulative LVR is a non-decreasing finite-variation process.
\end{theorem}

\begin{proof}
See Appendix~\ref{app:lvr-proof}.
\end{proof}

That is, the term $x_t\,dS_t$ represents the market risk exposure of the frictionless rebalancing strategy. This term is a local martingale and does not contribute to the predictable drift of the pool value; under the usual integrability conditions, it has zero conditional expectation. As such, all predictable losses are captured by the LVR component.

\begin{remark}\label{rem:depth-variance}
The LVR rate of Theorem~\ref{thm:lvr} factors into a liquidity term and a variance term. Define the local liquidity-depth term
$$D_t := \frac{S_tx_t}{S_t+r_t} \frac{1}{p_r(r_t,\pi_t)}$$
and the residual variance rate $Q_t := \frac{d}{dt}\langle\zeta\rangle_t$. Then $\ell_t = \frac{1}{2}D_tQ_t$. The term $D_t$ depends on the inventory position and the local curvature of the OP-AMM pricing rule, whereas $Q_t$ depends on the joint market and oracle innovations weighted by the oracle sensitivity.

Let $L_t^{\mathrm{OP}}/L_t$ denote the equivalent CPMM liquidity ratio from Lemma~\ref{lemma:liquidity}, evaluated relative to an equally capitalized CPMM at the current quoted price. Then $\frac{L_t^{\mathrm{OP}}}{L_t} = \frac{4S_t}{p_r(r_t,\pi_t)(S_t+r_t)^2}$. Since $V_t = x_t(S_t+r_t)$, substitution gives $D_t = \frac{V_t}{4} \frac{L_t^{\mathrm{OP}}}{L_t}$ and hence $\ell_t = \frac{V_t}{8} \frac{L_t^{\mathrm{OP}}}{L_t} Q_t$. Equivalently,
$$\frac{\ell_t}{V_t} = \frac{1}{8} \frac{L_t^{\mathrm{OP}}}{L_t} \frac{d}{dt}\langle\zeta\rangle_t.$$

This identity separates two effects of the oracle-parametrized design. Increasing the equivalent CPMM liquidity raises market depth but also increases the exposure to any residual price innovation. Therefore, oracle information improves overall performance only when the resulting reduction in $Q_t$ is sufficiently large to offset this increased depth.

This identity also isolates what the oracle itself contributes. For any information-agnostic design, i.e., $p_\pi\equiv0$, the residual variance rate is $Q_t=\sigma_t^2$ whatever the shape of the trading curve; the normalized LVR per unit of relative depth is, therefore, pinned at $\sigma_t^2/8$. That is, a concentrated liquidity design, e.g., Uniswap~V3~\cite{uniswapv3}, raises $L_t^{\mathrm{OP}}/L_t$ and raises $\ell_t/V_t$ in the same proportion, and no information-agnostic pricing rule can do otherwise. However, oracle information moves the AMM design off this line, and it does so exactly to the extent that it reduces $Q_t$ below $\sigma_t^2$. We wish to note, further, that a concentrated position retains its depth only while the external price remains within its range, whereas the depth of an OP-AMM is centered on the oracle price and so follows the external market without intervention.
\end{remark}

The residual variance rate $Q_t$ consists of the market volatility $\sigma_t^2$ (to which a standard AMM is exposed), the variation $p_\pi(r_t,\pi_t)^2(\sigma_t^\pi)^2$ imported through the oracle, and the covariance between these two innovations. In particular, the oracle term should not generally be interpreted as pure noise since a well-functioning oracle inherits variation from the underlying market itself. Standard AMMs rely on rebalancing arbitrage to restore prices endogenously. OP-AMMs, instead, incorporate external price information directly; in doing so, OP-AMMs can potentially reduce stale-price arbitrage while supplying deeper liquidity to liquidity-motivated order flow. We now specialize the LVR decomposition of Theorem~\ref{thm:lvr} to two oracles of particular interest, i.e., a perfect oracle in Example~\ref{ex:perfect} and a noisy oracle in Corollary~\ref{cor:noisy}.

\begin{example}\label{ex:perfect}
Consider an oracle that observes the true log-price exactly, i.e., $\pi_t=s_t=\log S_t$,
and normalize the initial pool value to $V_0=1$ for simplicity. As $d\pi_t=ds_t$, the residual price innovation satisfies $d\zeta_t = \bigl(1-p_\pi(r_t,s_t)\bigr)\,ds_t$ and therefore $d\langle\zeta\rangle_t/dt = \bigl(1-p_\pi(r_t,s_t)\bigr)^2\sigma_t^2$.
Consequently, the instantaneous LVR rate~\eqref{eq:LVR} reduces to
\begin{align}\label{eq:LVR-perfect}
\ell_t = \frac{S_tx_t}{2(S_t+r_t)p_r(r_t,s_t)} \bigl(1-p_\pi(r_t,s_t)\bigr)^2 \sigma_t^2.
\end{align}
Thus, under a perfect oracle, the residual price variation arises solely from the incomplete pass-through of the external price signal. We now wish to compare three OP-CFMM designs under this perfect oracle.
\begin{enumerate}
\item Consider first the baseline AMM, for which $p_\pi\equiv0$. In this case,
$$p_r^{\mathrm{AMM}}(r) = \frac{\phi'(r)}{\phi(r)}, \qquad
r_t^{\mathrm{AMM}} = \phi^{-1}(S_t).$$
For the CPMM, i.e., $\phi(r)=r$, we obtain
$$x_t^{\mathrm{AMM}} = \frac{1}{2\sqrt{S_0S_t}}, \qquad
V_t^{\mathrm{AMM}} = \sqrt{\frac{S_t}{S_0}}, \qquad
\ell_t^{\mathrm{AMM}} = \frac{1}{8}V_t^{\mathrm{AMM}}\sigma_t^2.$$
This recovers the constant-product specialization of the continuous-time LVR formula of~\cite{milionis2022automated}.

\item Consider now the GA of Example~\ref{ex:GAPMM} with fixed $\lambda\in(0,1)$ so that $p_\pi\equiv\lambda$. Under a perfect oracle, the clearing condition implies $\lambda s_t + (1-\lambda)\log\phi(r_t^{\mathrm{GA}}) = s_t$ and hence $r_t^{\mathrm{GA}} = \phi^{-1}(S_t)$. That is, the GA and the baseline AMM have the same equilibrium reserve ratio, but the GA has lower reserve sensitivity, $p_r^{\mathrm{GA}}(r,\pi) = (1-\lambda)\phi'(r)/\phi(r)$.

In the constant-product case,
$$x_t^{\mathrm{GA}} = \frac{1}{2\sqrt{S_0S_t}} \exp\left( \frac{\lambda}{8} \int_0^t\sigma_u^2\,du \right),$$
and hence
$$\ell_t^{\mathrm{GA}} = \frac{1-\lambda}{8} V_t^{\mathrm{GA}}\sigma_t^2, \qquad
\frac{V_t^{\mathrm{GA}}}{V_t^{\mathrm{AMM}}} = \exp\left( \frac{\lambda}{8} \int_0^t\sigma_u^2\,du \right).$$
In particular, the normalized LVR satisfies $\ell_t^{\mathrm{GA}}/V_t^{\mathrm{GA}} = (1-\lambda)\ell_t^{\mathrm{AMM}}/V_t^{\mathrm{AMM}}$ and is thus strictly decreasing in $\lambda$. The absolute instantaneous LVR, on the other hand, satisfies
$$\frac{\ell_t^{\mathrm{GA}}}{\ell_t^{\mathrm{AMM}}} = (1-\lambda) \exp\left( \frac{\lambda}{8} \int_0^t\sigma_u^2\,du \right),$$
so that $\ell_t^{\mathrm{GA}}<\ell_t^{\mathrm{AMM}}$ if and only if $\int_0^t\sigma_u^2\,du < -8\log(1-\lambda)/\lambda$.

\item Finally, the PT of Example~\ref{ex:PTPMM} is the full pass-through boundary case $p_\pi\equiv1$. It satisfies $d\zeta_t=0$ and $\ell_t\equiv0$. Hence, under a perfect oracle with complete pass-through, the pool incurs no systematic rebalancing loss.
\end{enumerate}
\end{example}

We wish to highlight two effects of oracle pass-through that are separated by these comparisons. Greater oracle sensitivity reduces the residual price variation that generates LVR, while also allowing the GA to retain more pool value over time. The first effect lowers the LVR per unit of pool value, whereas the second determines whether this improvement also holds in absolute terms.

We now extend the analysis to noisy oracle signals in the following corollary.

\begin{corollary}\label{cor:noisy}
Consider the setting of Theorem~\ref{thm:lvr} and assume that the oracle observes a noisy version of the true log-price, i.e.,
$$\pi_t=s_t+\eta_t,$$
where the oracle-error process $\eta_t$ follows an Ornstein--Uhlenbeck (OU) process
\begin{align}\label{eq:noisy-OU}
d\eta_t = -\theta \eta_t\,dt + \sigma_t^\eta\,dW_t^\eta, \qquad
d\langle W,W^\eta\rangle_t = \rho^\eta\,dt.
\end{align}
The parameter $\theta>0$ determines the rate of mean reversion, while $\sigma_t^\eta$ is the instantaneous innovation volatility of the oracle error.\footnote{Only the quadratic variation of the oracle error enters~\eqref{eq:LVR-noisy}. As such, this corollary holds for any continuous It\^o oracle error $d\eta_t = \mu_t^\eta\,dt + \sigma_t^\eta\,dW_t^\eta$. Furthermore, for the OU noise setting, for a given innovation volatility $\sigma_t^\eta$, the rate of mean reversion $\theta$ affects $\ell_t$ only indirectly through the distribution of $\eta_t$ and the resulting reserve state.} Then the instantaneous LVR rate is
\begin{align}\label{eq:LVR-noisy}
\ell_t = \frac{S_tx_t}{2(S_t+r_t)p_r(r_t,\pi_t)} \bigl( (1-p_\pi)^2\sigma_t^2 + p_\pi^2(\sigma_t^\eta)^2 - 2\rho^\eta p_\pi(1-p_\pi)\sigma_t\sigma_t^\eta \bigr),
\end{align}
where $p_\pi$ is evaluated at $(r_t,\pi_t)$ here and throughout the discussion below.
\end{corollary}

\begin{proof}
The oracle $\pi_t=s_t+\eta_t$ is a continuous It\^o process of the form considered in Section~\ref{sec:lvr-general}, with a generally time-varying correlation $\rho_t$ between its innovations and those of the market. As $d\pi_t = ds_t+d\eta_t$, the residual price innovation satisfies $d\zeta_t = \bigl(1-p_\pi(r_t,\pi_t)\bigr)\,ds_t - p_\pi(r_t,\pi_t)\,d\eta_t$. Consequently,
$$\frac{d}{dt}\langle\zeta\rangle_t = (1-p_\pi)^2\sigma_t^2 + p_\pi^2(\sigma_t^\eta)^2 - 2\rho^\eta p_\pi(1-p_\pi)\sigma_t\sigma_t^\eta$$
and the result follows from Theorem~\ref{thm:lvr}.
\end{proof}

When $\eta\equiv0$, i.e., $\eta_0=0$ and $\sigma_t^\eta\equiv0$, Corollary~\ref{cor:noisy} recovers the perfect oracle of Example~\ref{ex:perfect}. Within~\eqref{eq:LVR-noisy}, the first term captures the residual market-lag exposure, the second captures the oracle-error variation, and the third reflects the covariance between the market and oracle-error innovations. We wish to note that oracle errors also move the inventory $x_t$ away from that of a perfect-oracle pool. As $S_t$ is a local martingale, this deviation alters the variability of the pool value but not its predictable drift.

\begin{remark}\label{rem:stationary-variance}
Suppose that $\sigma_t^\eta\equiv\sigma^\eta$ is constant and that the OU process is in its stationary regime. The stationary variance of the oracle error is
$$v^\eta := \operatorname{Var}(\eta_t) = \frac{(\sigma^\eta)^2}{2\theta}.$$
If oracle noise is parametrized by fixing the stationary variance $v^\eta$, then maintaining the same long-run oracle-error variance requires $\sigma^\eta = \sqrt{2\theta v^\eta}$. Under this normalization, faster mean reversion requires larger innovation volatility and therefore increases both the instantaneous quadratic variation of the oracle error and its contribution to LVR in~\eqref{eq:LVR-noisy}.
\end{remark}

Equation~\eqref{eq:LVR-noisy} highlights the trade-off between lag risk and oracle-error risk. The oracle sensitivity $p_\pi$ determines how the residual variation is allocated between these two sources. As $p_\pi\to0$, the mechanism reduces to a standard AMM that is insensitive to oracle error; LVR is then driven entirely by market volatility. As $p_\pi\to1$, the lag-induced loss disappears, but the pool becomes fully exposed to the variation of $\eta_t$. Intermediate sensitivities balance these two sources of risk.

In the uncorrelated case $\rho^\eta=0$, the residual variance rate becomes $(1-p_\pi)^2\sigma_t^2 + p_\pi^2(\sigma_t^\eta)^2$. Therefore, increasing the oracle sensitivity reduces the exposure to market lags while increasing the exposure to oracle-error innovations. The relative magnitudes of $\sigma_t$ and $\sigma_t^\eta$ determine the desirable degree of oracle reliance. The following example makes this trade-off explicit for the GA.

\begin{example}\label{ex:GA-noisy}
Consider the GA of Example~\ref{ex:GAPMM} with fixed $\lambda\in(0,1)$ so that $p_\pi\equiv\lambda$. Then $p_r^{\mathrm{GA}}(r,\pi) = (1-\lambda)\phi'(r)/\phi(r) = (1-\lambda)p_r^{\mathrm{AMM}}(r)$, and hence $\frac{1}{p_r^{\mathrm{GA}}(r,\pi)} = \frac{1}{1-\lambda} \frac{1}{p_r^{\mathrm{AMM}}(r)}$. That is, at a given reserve ratio, the GA amplifies local liquidity depth by the factor $1/(1-\lambda)$. Under a perfect oracle, this increased depth is accompanied by a reduction in the residual market variation. Under a noisy oracle, however, the same mechanism also amplifies the exposure to oracle-error variation, which creates a trade-off between lag reduction and noise importation.

To isolate this trade-off, consider the uncorrelated case $\rho^\eta=0$ and compare the GA and baseline AMM locally while holding the state $(S_t,x_t,r_t)$ fixed. Under this ceteris paribus benchmark, the ratio of their instantaneous LVR coefficients is $\frac{(1-\lambda)^2\sigma_t^2 + \lambda^2(\sigma_t^\eta)^2}{(1-\lambda)\sigma_t^2}$. The GA coefficient is smaller than the baseline coefficient if and only if
$$\sigma_t^\eta < \bar{\sigma}_t^\eta(\lambda) := \sqrt{\frac{1-\lambda}{\lambda}}\,\sigma_t.$$
If $\sigma_t^\eta>\bar{\sigma}_t^\eta(\lambda)$, the amplification of oracle-error variation exceeds the reduction in lag-induced loss. Comparing instead against an information-agnostic design of the same relative depth, as in Remark~\ref{rem:depth-variance}, the common factor $1/(1-\lambda)$ drops out and the GA incurs the lower normalized LVR if and only if $\sigma_t^\eta < \sqrt{(2-\lambda)/\lambda}\,\sigma_t$. Notably, this threshold is strictly larger than $\bar{\sigma}_t^\eta(\lambda)$, i.e., the unconcentrated baseline is the more demanding comparison; under a perfect oracle, the GA attains a normalized LVR smaller than that of the equally deep information-agnostic design by the factor $(1-\lambda)^2$.

We wish to note that these thresholds are local, state-matched comparisons of the lag and noise coefficients. Under actual market clearing, the GA and the baseline AMM generally occupy different reserve and inventory states when the oracle is noisy. A full performance comparison must therefore also account for the differences in $(x_t,r_t)$ and local market depth as studied in Section~\ref{sec:cases}.
\end{example}

At the full pass-through boundary, the PT of Example~\ref{ex:PTPMM} satisfies $p_\pi\equiv1$. In this case, $d\zeta_t=-d\eta_t$ so that the lag-induced loss vanishes and LVR is driven entirely by the oracle-error variation; that is, each oracle-error innovation induces a corrective reserve adjustment. This is closest to the regime suggested by \cite[Section 10]{milionis2022automated}. Under a perfect oracle, it eliminates LVR entirely, as in Example~\ref{ex:perfect}; under any oracle error, however, that error is transmitted in full. As such, the oracle sensitivity is a design choice rather than a limit to approach.

\begin{remark}\label{rem:variance-fees}
The decomposition~\eqref{eq:LVR-noisy} isolates the oracle-error contribution $\frac{S_t x_t}{2(S_t+r_t)p_r(r_t,\pi_t)} p_\pi(r_t,\pi_t)^2(\sigma_t^\eta)^2$. This component scales with the instantaneous variance of the oracle error. Though fixed proportional fees may compensate liquidity providers for adverse-selection losses, their per-trade rate does not adapt directly to the quality of the oracle signal. This observation motivates a state-dependent spread that widens with the disagreement between the internal equilibrium price $\pi^*(r)$ of the OP-AMM and the oracle price $\pi$. As the leading oracle-error loss is driven by variance, a natural local specification is quadratic in the disagreement $\pi^*(r)-\pi$ so that its expected value is proportional to the oracle-error variance in the small-noise regime. The oracle-centered pricing of DODO~\cite{dodo} provides related practical motivation for this type of adaptive spread. As a complete analysis would also need to account for endogenous trading volume, fee revenue, and the interaction between fees and oracle sensitivity, we leave the design of such variance-linked fee schedules for future research.
\end{remark}

\subsection{Discrete Oracle Updates}\label{sec:lvr-stale}

We now consider an oracle with discrete updates. Let $(\tau_k)_{k\geq0}$ be a sequence of stopping times with $\tau_0:=0$ and $\tau_{k+1}>\tau_k$ almost surely. The oracle is piecewise constant, i.e.,
$$\pi_t=\pi_{\tau_k}, \qquad t\in[\tau_k,\tau_{k+1}),$$
with noisy updates
$$\pi_{\tau_k} = s_{\tau_k}+\eta_{\tau_k},$$
where $\eta_{\tau_k}$ represents measurement, aggregation, or update-time oracle error. This setting directly models a push oracle that remains frozen between update times and is refreshed only at $(\tau_k)$. In practice, an oracle update may contain both measurement error and posting latency. This can be modeled by $\pi_{\tau_k} = s_{\tau_k-\delta} + \eta_{\tau_k}$ for a delay $\delta>0$. The case study of Section~\ref{sec:sp500} incorporates such delay-induced oracle errors.

At an update time $t=\tau_k$, the oracle changes discontinuously from $\pi_{\tau_{k-1}}$ to $\pi_{\tau_k}$. As the external price process is continuous, the reserve ratio immediately before the update is $r_{\tau_k^-} = r^*(s_{\tau_k},\pi_{\tau_{k-1}})$, whereas, after the oracle update and subsequent arbitrage clearing, $r_{\tau_k} = r^*(s_{\tau_k},\pi_{\tau_k}) = r^*(s_{\tau_k},s_{\tau_k}+\eta_{\tau_k})$.

To describe this adjustment, let $x_k(r)$ denote the risky-asset reserve along the post-update conditional trading curve, defined by
$$U\bigl(x_k(r),r x_k(r),\pi_{\tau_k}\bigr) = U\bigl(x_{\tau_k^-},y_{\tau_k^-},\pi_{\tau_k}\bigr).$$
Then $x_k(r_{\tau_k^-})=x_{\tau_k^-}$, $x_k(r_{\tau_k})=x_{\tau_k}$, and $x_k'(r) = -\frac{x_k(r)}{r+\exp(p(r,\pi_{\tau_k}))}$. The following proposition provides the resulting decomposition of LVR under discrete oracle updates.

\begin{proposition}\label{prop:stale-lvr}
Consider the discrete oracle updates above and suppose that the conditions of Theorem~\ref{thm:lvr} hold between update times. Then, over any horizon $[0,T]$,
$$\mathrm{LVR}_T = \int_0^T\ell_t\,dt + \sum_{0<\tau_k\leq T} \Delta\mathrm{LVR}_{\tau_k},$$
where the continuous LVR rate on $(\tau_k,\tau_{k+1})$ is $\ell_t = \frac{1}{2} \frac{S_tx_t}{S_t+r_t} \frac{1}{p_r(r_t,\pi_{\tau_k})} \sigma_t^2$ and the jump contribution to LVR at each update time $\tau_k$ is
\begin{align}\label{eq:stale-jump-lvr}
\Delta\mathrm{LVR}_{\tau_k}:= -\Delta V_{\tau_k} = \int_{r_{\tau_k^-}}^{r_{\tau_k}} \Bigl( e^{p(r,\pi_{\tau_k})} - S_{\tau_k} \Bigr) x_k'(r)\,dr = S_{\tau_k} \int_{r_{\tau_k^-}}^{r_{\tau_k}} \Bigl( e^{p(r,\pi_{\tau_k})-s_{\tau_k}} - 1 \Bigr) x_k'(r)\,dr.
\end{align}
In particular, $\Delta\mathrm{LVR}_{\tau_k} \geq 0$ for every update time $\tau_k$.
\end{proposition}

\begin{proof}
See Appendix~\ref{app:proof-stale-lvr}.
\end{proof}

That is, the volatility-driven component between updates has the same form as in a standard AMM, but it is evaluated at the reserve state and liquidity depth induced by the most recent oracle update. As such, the oracle sensitivity is reflected through the prevailing inventory configuration and pricing geometry rather than through an instantaneous oracle innovation.

The jump term has a direct interpretation as the baseline \emph{oracle extractable value} (OEV)~\cite{andreoulis2025designing}. At $\tau_k$, the oracle changes while the reserves initially remain fixed at $(x_{\tau_k^-},y_{\tau_k^-})$. The refreshed quote therefore generally satisfies $p(r_{\tau_k^-},\pi_{\tau_k}) \neq s_{\tau_k}$, which creates an immediate arbitrage opportunity. Equation~\eqref{eq:stale-jump-lvr} measures the difference between the marginal execution price $e^{p(r,\pi_{\tau_k})}$ of the OP-AMM and the external price $S_{\tau_k}$, integrated over the reserve adjustment required to restore equilibrium. Thus, stale-oracle losses have two components: the \emph{continuous stale-pricing cost} $\int_0^T\ell_t\,dt$, which captures the losses accumulated between updates, and the \emph{oracle-refresh OEV} $\sum_{0<\tau_k\leq T}\Delta\mathrm{LVR}_{\tau_k}$, which captures the value extracted when the oracle is refreshed. This baseline jump loss assumes that arbitrage begins only after the update; allowing an attacker to trade both before and after the refresh can generate additional oracle-sandwich value as studied in Section~\ref{sec:lvr-oev}. The following example makes the stale-oracle effect explicit for the GA.

\begin{example}\label{ex:GA-stale}
Consider the GA of Example~\ref{ex:GAPMM} with fixed $\lambda\in(0,1)$. Between update times, the clearing condition becomes $\lambda\pi_{\tau_k} + (1-\lambda)\log\phi(r_t^{\mathrm{GA}}) = s_t$ and hence
$$r_t^{\mathrm{GA}} = \phi^{-1}\left( \exp\left( \frac{s_t-\lambda\pi_{\tau_k}}{1-\lambda} \right) \right).$$
Unlike the perfect-oracle benchmark of Example~\ref{ex:perfect}, the equilibrium reserve ratio now depends explicitly on the discrepancy between the current market price and the most recent oracle update. A stale oracle, therefore, induces a persistent displacement of the reserve state between updates. As $p_r^{\mathrm{GA}}(r,\pi) = (1-\lambda)\phi'(r)/\phi(r) = (1-\lambda)p_r^{\mathrm{AMM}}(r)$, the continuous LVR rate is
$$\ell_t^{\mathrm{GA}} = \frac{1}{2} \frac{1}{1-\lambda} \frac{S_tx_t^{\mathrm{GA}}}{S_t+r_t^{\mathrm{GA}}} \frac{1}{p_r^{\mathrm{AMM}}(r_t^{\mathrm{GA}})} \sigma_t^2.$$
Thus, the oracle sensitivity amplifies local liquidity depth through the factor $1/(1-\lambda)$, but around a reserve state determined by the stale oracle signal rather than the contemporaneous market price. This displacement is resolved at update times through the jump term $\Delta\mathrm{LVR}_{\tau_k}$.
\end{example}

To conclude our discussion of stale oracles, we wish to consider the choice of the oracle update policy. The stale-oracle decomposition suggests that such update policies can be optimized through the choice of stopping times. For example, given a candidate oracle signal $\widetilde{s}_t$, e.g., the true log-price $s_t$, a noisy signal $s_t+\eta_t$, or a delayed signal $s_{t-\delta}$, one may consider
\begin{equation}\label{eq:update-policy}
\tau_{k+1} = \Bigl( \inf\bigl\{t>\tau_k \mid |\widetilde{s}_t-\pi_{\tau_k}| > \varepsilon \bigr\} \Bigr) \wedge (\tau_k+h)
\end{equation}
with $\pi_{\tau_{k+1}} = \widetilde{s}_{\tau_{k+1}}$, where $\varepsilon>0$ is a deviation threshold and $h>0$ is a maximal refresh interval, or heartbeat.\footnote{In Section~\ref{sec:sp500} we abuse notation and let $\varepsilon = 0$ denote an update at every observation.}
Choosing $(\varepsilon,h)$ balances the continuous and jump components of stale-oracle LVR against the operational cost of oracle refreshes. We leave the optimization of such update policies for future research.

\begin{remark}\label{rem:decaying-sensitivity}
The characterization of Theorem~\ref{thm:interp} treats the oracle sensitivity as state dependent; however, the oracle sensitivity may also vary with the age of the oracle signal. Let $\Delta t := t-\tau_k$ denote the time elapsed since the most recent oracle update and consider a sensitivity $\lambda=\lambda(\Delta t)$ satisfying
$$0 \leq \lambda(\Delta t) < 1, \qquad
\lambda'(\Delta t) < 0, \qquad
\lim_{\Delta t\to\infty} \lambda(\Delta t) = 0,$$
with $\lambda(0)$ close to one. Economically, this describes a transition from price taking to price making. Immediately after an oracle update, the OP-AMM places substantial weight on the external signal so as to minimize price lag. As the signal ages, and its error variance accumulates with market volatility, the OP-AMM gradually shifts weight toward its internal equilibrium price and ultimately relies on endogenous price discovery. In this way, the time-decaying sensitivity acts as a soft circuit breaker that can protect liquidity providers from stale-quote arbitrage during periods of oracle latency or outage.
\end{remark}

\subsubsection{Oracle-Update Sandwich Attacks}\label{sec:lvr-oev}

The analysis of Section~\ref{sec:lvr-stale} assumes that arbitrageurs trade against the OP-AMM only after each oracle update so as to restore $p(r_t,\pi_t)=s_t$; these trades generate the oracle-refresh loss $\Delta\mathrm{LVR}_{\tau_k}$ of Proposition~\ref{prop:stale-lvr}. On a blockchain, however, a trader can also swap against the pool immediately before an update is posted and reverse that swap immediately afterward, i.e., a sandwich attack~\cite{zhou2021high,qin2022quantifying}. We refer to such a trader as an \emph{attacker} since, unlike the arbitrageurs of Section~\ref{sec:lvr-stale}, who trade according to the pricing rule of the pool, an attacker exploits the fact that the oracle update is posted as a separate transaction. Within this section, we will quantify the additional value that an attacker can extract as well as the protection that transaction fees provide against it. We first compute this value without fees in Proposition~\ref{prop:MEV-0fee}, compare it with the oracle-refresh loss in Corollary~\ref{cor:lvr-no-attack}, and then introduce fees in Lemma~\ref{lemma:MEV+fees}.

Consider a single oracle update
$$\pi_-\longrightarrow\pi_+, \qquad \dpi:=\pi_+-\pi_-\neq0,$$
and hold the external price $S=e^s$ fixed over the update window. Let $(x_-,y_-)$ denote the pre-update reserves of the OP-AMM and set $r_- = y_-/x_-$.

Throughout this section we will work in the geometry of Proposition~\ref{prop:CI}. Recall from Section~\ref{sec:lvr-general} that strict monotonicity and surjectivity of $p(\cdot,\pi)$ determine a unique clearing reserve ratio
$$r^*(s,\pi) := p(\cdot,\pi)^{-1}(s), \qquad (s,\pi) \in \R^2.$$
As the canonical trading function $\lcal(\cdot,\cdot,\pi) = \exp U(\cdot,\cdot,\pi)$ of Proposition~\ref{prop:CI} is positively homogeneous, we refer to $\lcal(x,y,\pi)$ as the \emph{invariant liquidity} of the reserve state $(x,y)$ at oracle price $\pi$ in the sense of~\cite[Section 1.6]{angeris2023geometry}.\footnote{The invariant liquidity should be distinguished from the equivalent CPMM liquidity $L^{\mathrm{OP}}$ of Lemma~\ref{lemma:liquidity}. The former is a global, positively homogeneous scale factor on the reachable set and is constant along a conditional trading curve, whereas the latter is a local curvature match at a single state and varies along that curve. These notions coincide for the CPMM, for which the curvature match is exact, but not in general; e.g., for the geometric-average family of Example~\ref{ex:GAPMM}, their ratio varies with the state.} Following~\cite[Section 1.3]{angeris2023geometry}, a trade against the pool at oracle price $\pi$ can carry the reserves from $(x,y)$ to any state of the rescaled reachable set $\lcal(x,y,\pi)\rcal(\pi)$, whose efficient boundary is the conditional trading curve through $(x,y)$. Write $\lcal_- := \lcal(x_-,y_-,\pi_-)$. We refer to the efficient boundary of $\lcal_-\rcal(\pi_-)$ as the \emph{stale curve} and, for any reserve state $(x',y')$, to the efficient boundary of $\lcal(x',y',\pi_+)\rcal(\pi_+)$ as the \emph{refreshed curve} through $(x',y')$.

In this geometry, marking a pool to market at the external price is an evaluation of the portfolio value function~\cite[Section 1.4.2]{angeris2023geometry}
$$\V_\pi(S) := \min\bigl\{Sx+y \mid (x,y) \in \rcal(\pi)\bigr\}, \qquad S>0.$$
As $\rcal(\pi)$ is closed and upward closed, this minimum is attained at the unique point of $\rcal(\pi)$ with reserve ratio $\bar r:=r^*(\log S,\pi)$ so that
$$\V_\pi(S) = (S+\bar r)\exp\bigl(-G(\bar r,\pi)\bigr).$$

Finally, by the boundary limits~\eqref{eq:g-limits}, $U(x,0,\pi) = \log x + g_0(\pi)$ and $U(0,y,\pi) = \log y + g_\infty(\pi)$. Therefore, $\rcal(\pi)$ meets the $y$-axis, i.e., a finite trade can drain the risky reserve, if and only if $g_\infty(\pi)>-\infty$. Likewise, $\rcal(\pi)$ meets the $x$-axis, i.e., a finite trade can drain the num\'eraire reserve, if and only if $g_0(\pi)>-\infty$. Boundary divergence, i.e., the structural property~(v) of Section~\ref{sec:pmm-construction}, excludes both so that $\rcal(\pi)\subseteq\R^2_{++}$.

The attacker first trades against the stale curve, moving the reserves to some $(x',y')\in\lcal_-\rcal(\pi_-)$ before the update, i.e., a \emph{front-run}. After the update, the attacker trades against the refreshed curve through $(x',y')$, i.e., a \emph{back-run}. For a front-run to $(x',y')$, the largest value that the attacker can extract from this round trip is
$$\Pi(x',y') := \sup \bigl\{S(x_--x_+)+(y_--y_+) \mid (x_+,y_+) \in \lcal(x',y',\pi_+)\rcal(\pi_+) \bigr\},$$
and the supremal sandwich value is
$$\Pi^* := \sup_{(x',y')\in\lcal_-\rcal(\pi_-)} \Pi(x',y').$$

The following proposition provides the value of a sandwich attack around an oracle update in the absence of fees.

\begin{proposition}\label{prop:MEV-0fee}
Consider the single oracle update above.
\begin{enumerate}
\item\label{prop:MEV-backrun}
For every $(x',y')\in\lcal_-\rcal(\pi_-)$, the optimal back-run attains the supremum at the unique point of reserve ratio $r_+ := r^*(s,\pi_+)$ and
$$\Pi(x',y') = \bigl[Sx_-+y_-\bigr] - \lcal(x',y',\pi_+)\,\V_{\pi_+}(S).$$

\item\label{prop:MEV-sup}
If $\dpi>0$, then the supremal value is
$$\Pi^* = \bigl[Sx_-+y_-\bigr] - \lcal_-e^{D}\,\V_{\pi_+}(S), \qquad D := \lim_{r \nearrow \infty}\bigl[G(r,\pi_+)-G(r,\pi_-)\bigr] \in [-\infty,0],$$
with the convention $e^{-\infty}:=0$. If $g_\infty(\pi_-)>-\infty$, then $D = g_\infty(\pi_+)-g_\infty(\pi_-)$ and the supremum is attained. In particular, an optimal front-run is the finite trade carrying the pool to the boundary point $\bigl(0,\lcal_-e^{-g_\infty(\pi_-)}\bigr) \in \lcal_-\rcal(\pi_-)$, which drains the risky reserve entirely. If $g_\infty(\pi_-)=-\infty$, then the supremum is approached along the stale curve as $r\nearrow\infty$; if, in addition, $p(r,\pi_+)>p(r,\pi_-)$ for every $r>0$, then it is not attained at any finite reserve state. In either case, $\Pi^*$ equals the entire pool value $Sx_-+y_-$ if and only if $D=-\infty$. The case $\dpi<0$ is symmetric with $g_0$ and $r \searrow 0$ in place of $g_\infty$ and $r \nearrow \infty$.
\end{enumerate}
\end{proposition}

\begin{proof}
See Appendix~\ref{app:proof-MEV-0fee}.
\end{proof}

That is, by Proposition~\ref{prop:MEV-0fee}\eqref{prop:MEV-backrun}, the attacker extracts the value of the pool less the value of the refreshed reachable set, scaled by the invariant liquidity that the update leaves behind. Furthermore, by Proposition~\ref{prop:MEV-0fee}\eqref{prop:MEV-sup}, absent fees, capital constraints, or reserve limits, the attacker front-runs as far as possible in the direction of the oracle update, i.e., toward draining the risky reserve when $\dpi>0$.

Choosing $(x',y')=(x_-,y_-)$ corresponds to no front-run before the oracle update. The following corollary relates this choice to the oracle-refresh jump of Proposition~\ref{prop:stale-lvr}.

\begin{corollary}\label{cor:lvr-no-attack}
Consider an oracle update at an update time $\tau=\tau_k$ of Section~\ref{sec:lvr-stale}, i.e., $\pi_-=\pi_{\tau_{k-1}}$, $\pi_+=\pi_{\tau_k}$, $S=S_{\tau_k}$, and $(x_-,y_-)=(x_{\tau_k^-},y_{\tau_k^-})$, so that the pool is at its pre-update equilibrium $p(r_-,\pi_-)=s$. Then the frictionless value without a front-run is the oracle-refresh jump of Proposition~\ref{prop:stale-lvr}, i.e., $\Pi(x_-,y_-)=\Delta\mathrm{LVR}_\tau$. In particular, $\Pi^*\geq\Delta\mathrm{LVR}_\tau$.
\end{corollary}
\begin{proof}
The refreshed curve through $(x_-,y_-)$ is the post-update curve of~\eqref{eq:stale-jump-lvr}, which carries the pool value $V(r)=x_k(r)(S+r)$ with $V(r_{\tau^-}) = Sx_-+y_-$ and, by Proposition~\ref{prop:MEV-0fee}\eqref{prop:MEV-backrun}, $V(r_\tau) = \lcal(x_-,y_-,\pi_+)\V_{\pi_+}(S)$. Hence
$$\Delta\mathrm{LVR}_{\tau} = -\Delta V_{\tau} = \bigl[Sx_-+y_-\bigr] - \lcal(x_-,y_-,\pi_+)\,\V_{\pi_+}(S) = \Pi(x_-,y_-)$$
and the result follows as $\Pi^*\geq\Pi(x_-,y_-)$ by definition.
\end{proof}

The inequality of Corollary~\ref{cor:lvr-no-attack} may be strict because the attacker can deliberately move the pool along the stale curve before the update.

\begin{remark}
\label{rem:MEV-attainment}
Both regimes of Proposition~\ref{prop:MEV-0fee}\eqref{prop:MEV-sup} are realized by the constructions of Section~\ref{sec:pmm}. Consider the CPMM baseline $\phi(r)=r$ and a pool at $\pi_-=0$ with $r_-=1$ and $x_-=y_-=1$, an external price $S=1$ (so that the pool value is $Sx_-+y_-=2$), and an oracle refresh to $\pi_+=\frac{1}{2}$.

For the price-tracking family of Example~\ref{ex:PTPMM}, $G(r,\pi) = \frac{\log r}{1+e^\pi}$; therefore, as boundary divergence holds, $D=-\infty$ and $\Pi^*=2$. In this case, the frictionless sandwich extracts the entire pool, but only in the limit of an unbounded trade. In contrast, for the geometric-average family of Example~\ref{ex:GAPMM} with $\lambda=\frac{3}{5}$, $g_\infty(\pi) = -\frac{1}{\lambda}\log\bigl(1+e^{\lambda\pi}\bigr)$ is finite and boundary divergence fails. As such, the supremum is attained by a finite trade. Specifically, the optimal front-run buys the entire risky reserve for $2.175$ units of the num\'eraire, carrying the pool to $(0,3.175)$; this yields $\Pi^*\approx0.512$, i.e., approximately $25.6\%$ of the pool value, as compared to $\Delta\mathrm{LVR}_\tau\approx0.054$ without a front-run.

Notably, this distinction follows directly from the structural properties of these two designs. Boundary divergence is exactly what prevents the pool from being emptied by a finite trade, and the geometric-average design does not satisfy this property.
\end{remark}

The frictionless setting of Proposition~\ref{prop:MEV-0fee} isolates the theoretical value created by the oracle update. We now introduce a proportional fee $\gamma\in[0,1)$ on the asset sold to the pool and assume that fees are held in escrow for liquidity providers rather than reinvested in the trading reserves, as in Uniswap~V3~\cite{uniswapv3}. Let $\Gamma := -\log(1-\gamma) \geq0$. Fees affect the attack through two channels. First, arbitrage confines the quote to the no-trade region $p(r_t,\pi_t) \in [s_t-\Gamma,s_t+\Gamma]$ rather than pinning it exactly to $s_t$. Second, the attacker pays the fee on both legs of the round trip. We refer the interested reader to~\cite{sadeghi2026liquidation} for a study of the protection that fees provide against sandwich attacks in a different setting.

For a positive oracle update $\dpi>0$, Proposition~\ref{prop:MEV-0fee}\eqref{prop:MEV-sup} shows that, in the frictionless setting, the supremal value is approached by front-runs that increase the reserve ratio (i.e., decrease the risky-asset reserve). As such, we restrict the fee-adjusted analysis to
$$\ccal_- := \bigl\{(x,y)\in\lcal_-\rcal(\pi_-) \mid x\leq x_- \bigr\}.$$
As $U(\cdot,\cdot,\pi_-)$ is strictly increasing in each argument, $x\leq x_-$ together with $U(x,y,\pi_-)\geq\log\lcal_-$ forces $y\geq y_-$ on $\ccal_-$. The front-run therefore buys the risky asset from the pool and the back-run sells it back, so that the fee is paid in the num\'eraire on the front-run and in the risky asset on the back-run. The case $\dpi<0$ is obtained by reversing the trade directions throughout. The following lemma demonstrates that transaction fees can eliminate the incremental value of front-running an oracle update.

\begin{lemma}\label{lemma:MEV+fees}
Assume $\dpi>0$ and let the pre-update pool satisfy $p(r_-,\pi_-) \in [s-\Gamma,s+\Gamma]$. For $(x',y')\in\ccal_-$, define
\begin{align*}
\Pi_\gamma(x',y') := \sup \Biggl\{\begin{array}{l}
S\Bigl[(x_- - x') + \frac{x' - x_+}{1-\gamma}\Bigr] \\
\qquad + \Bigl[\frac{y_- - y'}{1-\gamma} + (y' - y_+)\Bigr]
\end{array} \Biggm| (x_+,y_+) \in \lcal(x',y',\pi_+)\rcal(\pi_+),\; x_+ \geq x'\Biggr\}
\end{align*}
and let $\Pi_\gamma^*:= \sup_{(x',y')\in\ccal_-} \Pi_\gamma(x',y')$. Suppose that the oracle-induced quote displacement satisfies
\begin{align}\label{eq:fee-quote-shift}
p(r,\pi_+)-p(r,\pi_-) \leq \Gamma, \qquad r\geq r_-.
\end{align}
Then $\Pi_\gamma^*=\Pi_\gamma(x_-,y_-)$, i.e., within $\ccal_-$, front-running the oracle update creates no additional value beyond the post-update arbitrage opportunity. If, additionally, $p(r_-,\pi_+) \leq s+\Gamma$, then $\Pi_\gamma^* = 0$.
\end{lemma}

\begin{proof}
See Appendix~\ref{app:proof-MEV-fees}.
\end{proof}

That is, front-running an oracle update creates no incremental value whenever the oracle-induced quote displacement is no larger than the fee band, i.e., whenever~\eqref{eq:fee-quote-shift} holds. If the refreshed quote also remains inside the post-update no-trade region, then the entire post-update arbitrage opportunity disappears.

\begin{remark}
\label{rem:MEV-mitigation}
By non-expansiveness of the oracle response, $p(r,\pi_+)-p(r,\pi_-) \leq |\dpi|$ so that the simpler condition $|\dpi| \leq \Gamma$ is sufficient for~\eqref{eq:fee-quote-shift}. For the GA, the quote displacement is $p_{\mathrm{GA}}(r,\pi_+) - p_{\mathrm{GA}}(r,\pi_-) = \lambda\dpi$ and the corresponding protection condition becomes $\lambda|\dpi| \leq \Gamma$.

Consider a continuous candidate signal that is updated immediately at its first threshold crossing, i.e., with no publication delay or overshoot. For such a signal, choosing the deviation threshold of~\eqref{eq:update-policy} to satisfy $\varepsilon \leq \Gamma$ guarantees that deviation-triggered updates satisfy $|\dpi|\leq\Gamma$. For the GA, the analogous sufficient condition is $\lambda\varepsilon \leq \Gamma$. However, heartbeat updates, delayed publication, discrete sampling, or jumps may produce larger updates and so require separate controls.

Additional mechanisms can reduce the ability of an attacker to atomically bracket an oracle update:
\begin{enumerate}
\item \emph{Atomic pull-oracle consumption}: A fresh, verified oracle report can be supplied and consumed within the same transaction; this eliminates the separately observable on-chain update that could otherwise be sandwiched. The protection provided depends on the freshness of the report and on transaction-ordering guarantees.

\item \emph{Priority execution}: Placing the oracle update before ordinary swaps in the same block prevents a same-block front-run against the stale curve, though cross-block positioning may remain possible.

\item \emph{Delayed or committed activation}: Activating a committed oracle value under controlled ordering can prevent an atomic sandwich of the update. This protection may introduce additional latency and, therefore, trades off the mitigation of maximal extractable value (MEV) against stale-price risk.
\end{enumerate}
\end{remark}

\section{Case Studies}\label{sec:cases}

The preceding sections develop two complementary performance measures for OP-AMMs: local capital efficiency, measured by the equivalent CPMM liquidity under equal pool value, and adverse-selection losses, measured by LVR. Within this section, we use these measures to quantify when and how the CPMM-based GA improves upon the benchmark CPMM. Section~\ref{sec:frontier} provides an analytical comparison of these measures under a noisy oracle, while Section~\ref{sec:sp500} conducts a counterfactual backtest using historical S\&P~500 market data.

\subsection{Relative Market Depth and Arbitrage Costs for Noisy Oracles}
\label{sec:frontier}

Consider the CPMM-based GA of Example~\ref{ex:GAPMM}. We will vary the oracle sensitivity $\lambda\in[0,1)$, where $\lambda=0$ recovers the information-agnostic CPMM and $\lambda\nearrow 1$ approaches a fully oracle-dependent design. Throughout this section, we consider the noisy-oracle setting of Corollary~\ref{cor:noisy} with constant volatilities and uncorrelated innovations $(\rho^\eta=0)$. Let $\alpha:=\sigma^\eta/\sigma$ denote the relative oracle noise. We parametrize the initial oracle misalignment by
$$|\log S_0-\pi_0| = k\frac{\sigma^\eta}{\sqrt{2\theta}},$$
where $k\geq0$ measures the misalignment in units of the stationary standard deviation of the oracle error. As the performance measures below are symmetric in the sign of $\log S_0-\pi_0$, we take $\log S_0-\pi_0\geq0$ without loss of generality. Further, we normalize $\sigma=1$ and $\theta=1/2$ so that $\log S_0-\pi_0=k\alpha$. In order to evaluate this OP-CFMM, we will consider two metrics over varying oracle sensitivities $\lambda \in [0,1)$:
\begin{itemize}
\item To measure capital efficiency, we consider the relative market depth at the realized external price (see Example~\ref{ex:GA-liquidity}):
$$\frac{L_{\mathrm{GA}}}{L} = \frac{1}{1-\lambda} \frac{4m_\lambda}{(1+m_\lambda)^2}, \qquad
m_\lambda = \exp\left( \frac{\lambda}{1-\lambda}k\alpha \right).$$
Larger values correspond to smaller local price impacts for traders.

\item To measure continuous adverse-selection losses, we consider the relative instantaneous LVR rate under the same pool-value normalization:
$$\frac{\ell_{\mathrm{GA}}}{\ell} = \bigl( (1-\lambda)^2+\lambda^2\alpha^2 \bigr) \frac{L_{\mathrm{GA}}}{L}.$$
Smaller values correspond to lower continuous LVR relative to the benchmark CPMM.
\end{itemize}
We take these relative values so as to directly compare the OP-CFMM with the benchmark CPMM. If the relative market depth is greater than $1$ and the relative arbitrage loss is less than $1$, then the OP-CFMM strictly dominates the information-agnostic AMM in these two metrics.

For both metrics, we first characterize their behavior at $\lambda=0$ and as $\lambda\nearrow1$. By construction, both the relative market depth and the relative arbitrage loss equal $1$ at $\lambda=0$, i.e., the GA coincides with the benchmark exactly. Further, taking the derivatives of these metrics at $\lambda=0$, we find that the relative market depth is initially increasing with $\left. \partial_\lambda \frac{L_{\mathrm{GA}}}{L} \right|_{\lambda=0} = 1$ while the relative arbitrage loss is initially decreasing with $\left. \partial_\lambda \frac{\ell_{\mathrm{GA}}}{\ell} \right|_{\lambda=0} = -1$. Therefore, there exists some sufficiently small $\lambda_*\in(0,1)$ such that the OP-CFMM strictly dominates the CPMM in both metrics simultaneously for every $\lambda\in(0,\lambda_*)$.

However, as $\lambda\nearrow1$, we need to distinguish between three cases. First, if $\alpha=0$, then the oracle is noise-free; the relative market depth diverges to $\infty$ while the relative arbitrage loss converges to $0$. That is, without oracle noise, increasing the oracle sensitivity improves performance in both metrics monotonically. Second, if $k=0$ but $\alpha>0$, then the oracle is currently aligned with the market but remains noisy. In this case, both the relative market depth and the relative arbitrage loss diverge as $\lambda\nearrow1$. Finally, if $k\alpha>0$, then both metrics converge to $0$ as $\lambda\nearrow1$. The vanishing relative arbitrage loss does not represent dominance since the relative market depth at the external price also collapses. Consequently, there exists some $\lambda^*\in(0,1)$ such that the OP-CFMM does not dominate the CPMM for any $\lambda\in(\lambda^*,1)$.

These theoretical behaviors are demonstrated numerically in Figures~\ref{fig:theoretical-k=0.5}--\ref{fig:theoretical-k=2.0}, which plot the relative market depth against the relative arbitrage loss for varying combinations of oracle error $(k)$ and noise $(\alpha)$. Figure~\ref{fig:theoretical-lambda} explicitly maps the oracle sensitivities for which the GA weakly dominates the CPMM, i.e., where the GA simultaneously provides greater capital efficiency $(L_{\mathrm{GA}}/L\geq1)$ and incurs lower adverse selection $(\ell_{\mathrm{GA}}/\ell\leq1)$. Reading this plot from $\lambda=0$ upward, the initial boundary crossing indicates a shift from dominance to non-dominance, though not necessarily strict underperformance relative to the CPMM; subsequent re-entries into the dominance region highlight the non-trivial nature of this parameter space. As such, the oracle sensitivity must be carefully calibrated, e.g., through backtesting as in Section~\ref{sec:sp500}.

\begin{figure}[t]
\centering
\begin{subfigure}[t]{0.45\textwidth}
\centering
\includegraphics[width=\textwidth]{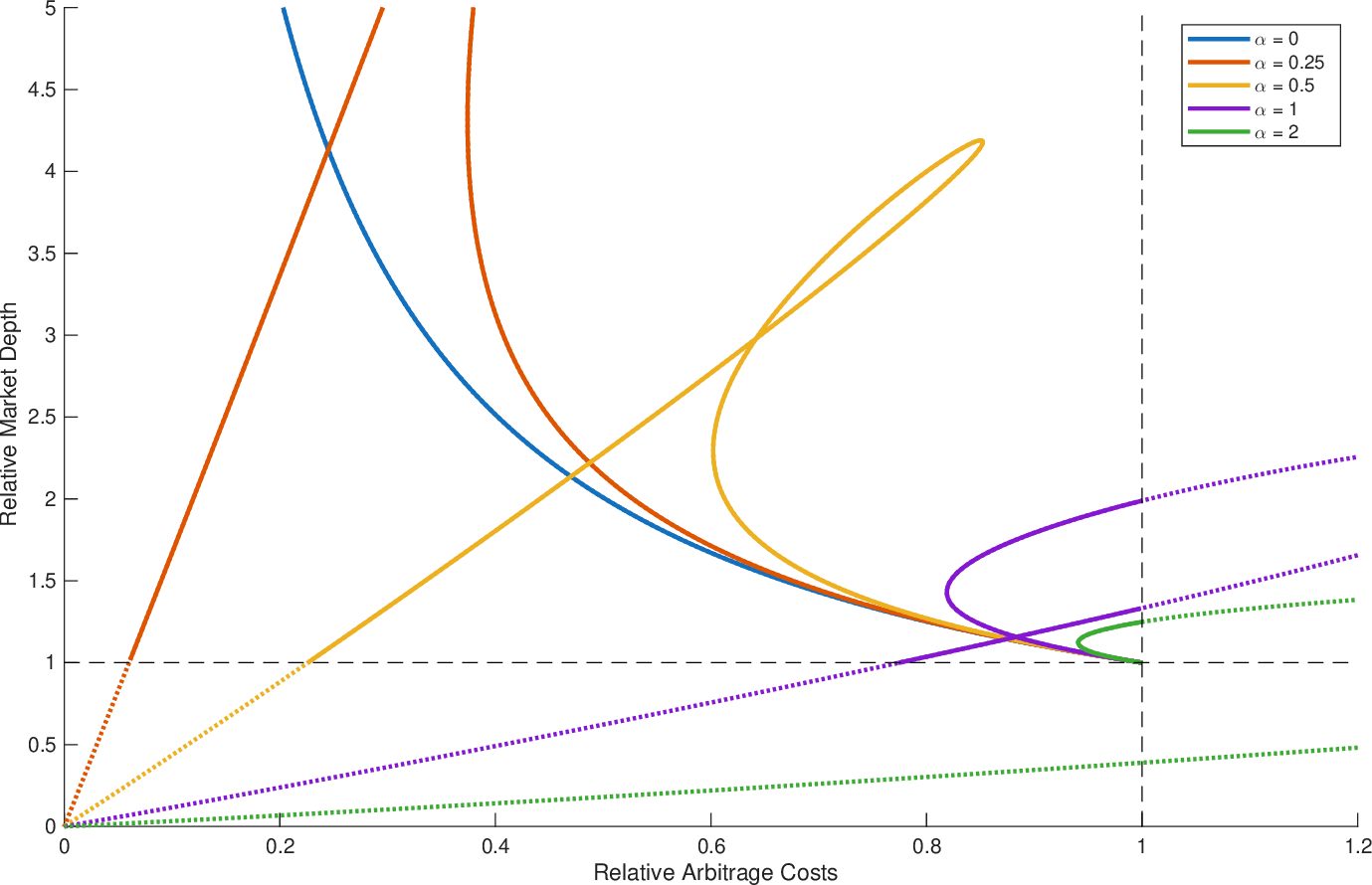}
\caption{$k = 0.5$ standard deviations of oracle error.}
\label{fig:theoretical-k=0.5}
\end{subfigure}
~~~~~
\begin{subfigure}[t]{0.45\textwidth}
\centering
\includegraphics[width=\textwidth]{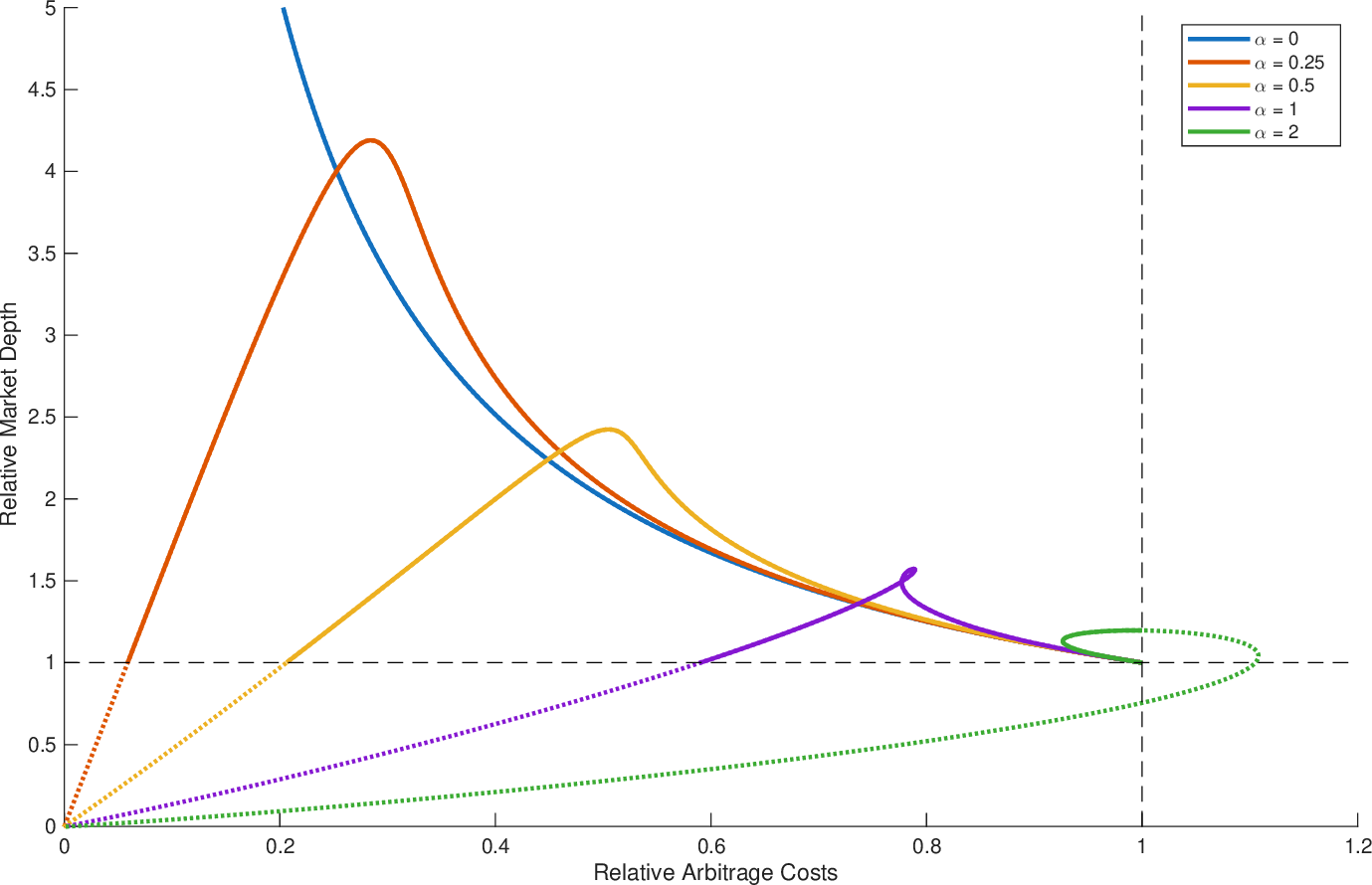}
\caption{$k = 1.0$ standard deviations of oracle error.}
\label{fig:theoretical-k=1.0}
\end{subfigure}
~~~~~
\begin{subfigure}[t]{0.45\textwidth}
\centering
\includegraphics[width=\textwidth]{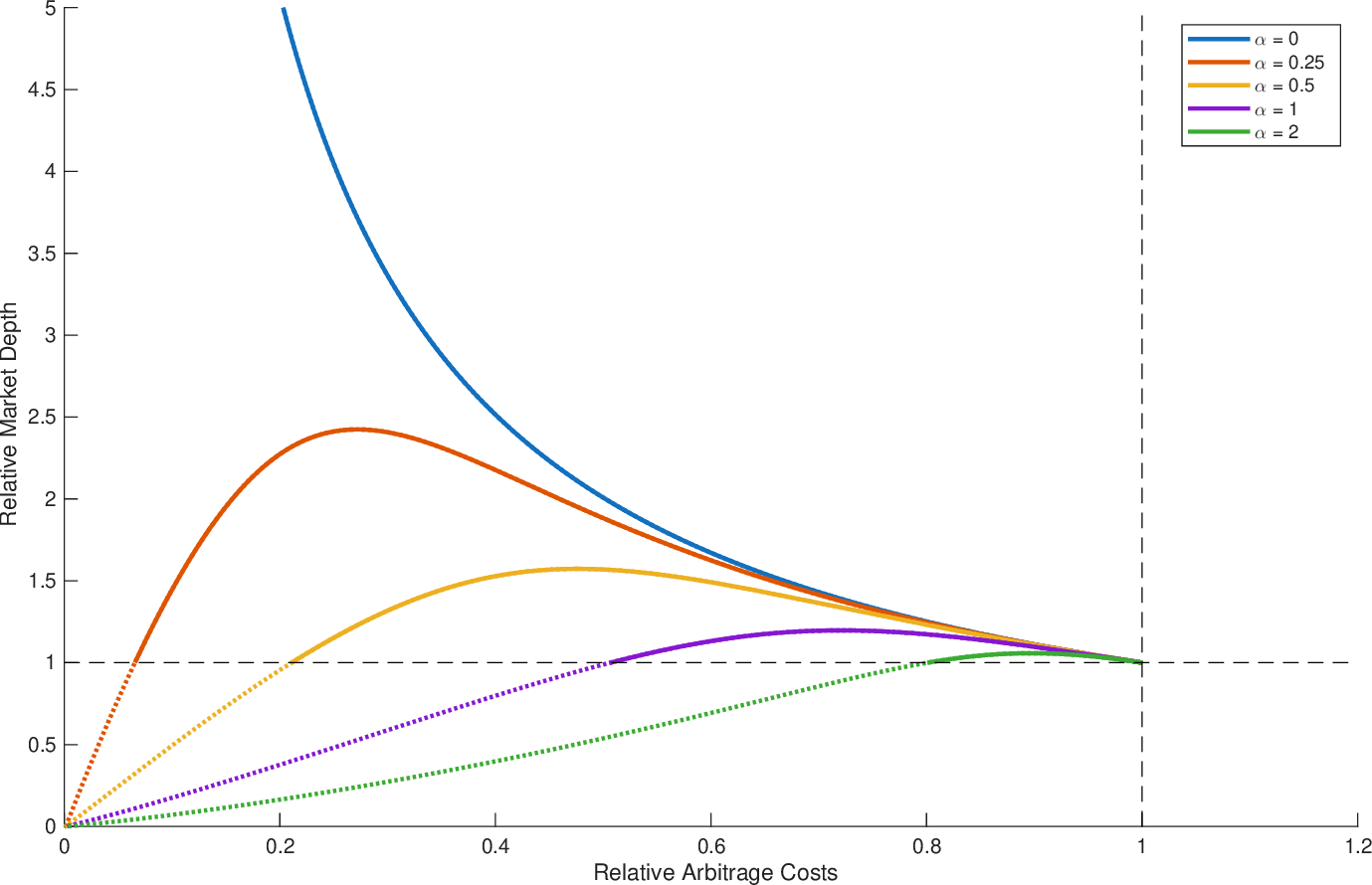}
\caption{$k = 2.0$ standard deviations of oracle error.}
\label{fig:theoretical-k=2.0}
\end{subfigure}
~~~~
\begin{subfigure}[t]{0.45\textwidth}
\centering
\includegraphics[width=\textwidth]{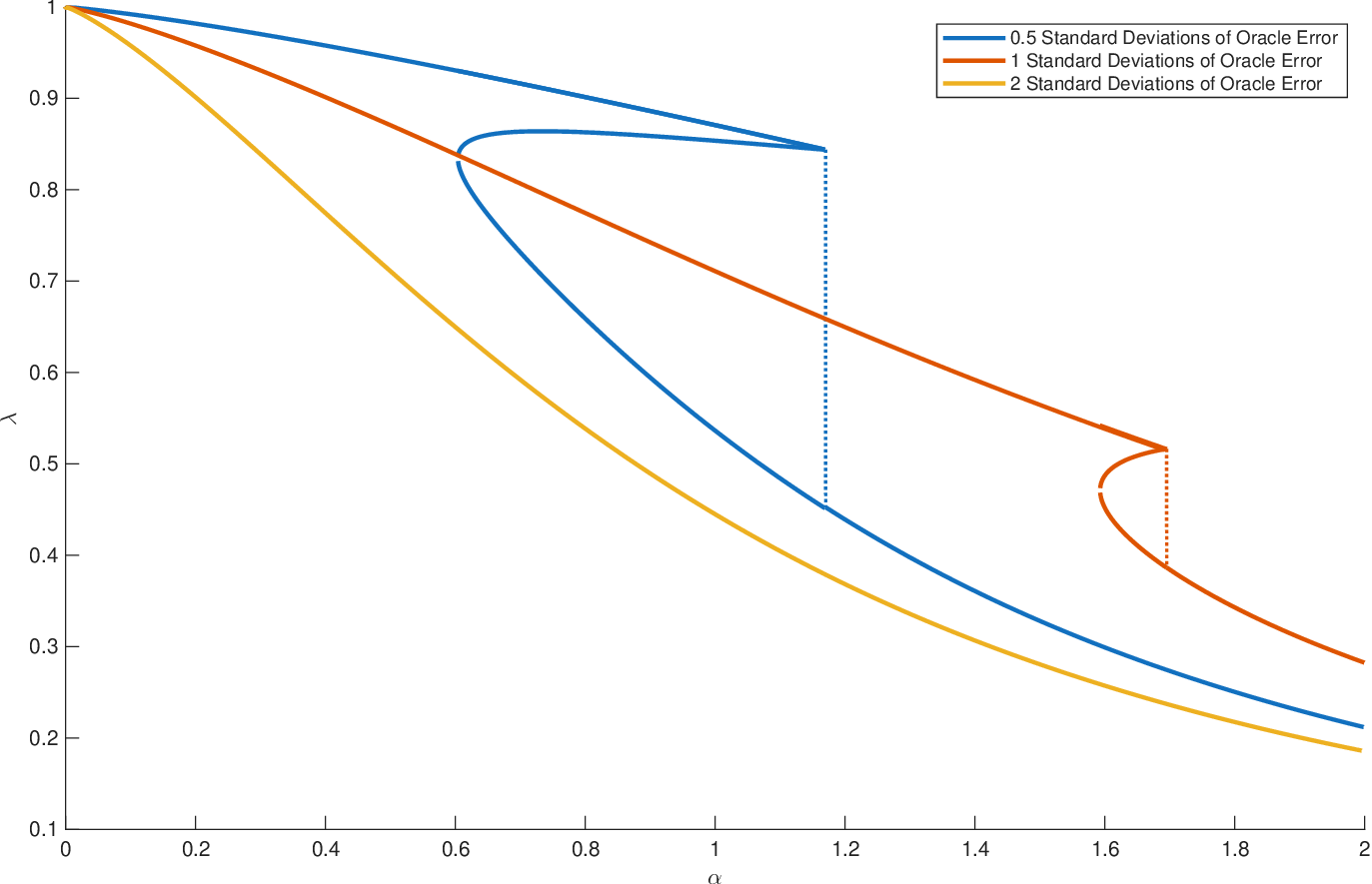}
\caption{Oracle sensitivities ($\lambda$) dominating the CPMM.}
\label{fig:theoretical-lambda}
\end{subfigure}
\caption{Section~\ref{sec:frontier}: Trade-off between market depth and losses to arbitrageurs for the GA compared to the benchmark CPMM.}
\label{fig:theoretical}
\end{figure}

\subsection{S\&P~500 Case Study}\label{sec:sp500}

Consider now the performance of the GA in a counterfactual backtest. For this purpose, we simulate a tokenized SPY market using one-second NBBO data for this ETF, which tracks the S\&P~500 index, over the 2023 calendar year. Though arbitrageurs observe the contemporaneous NBBO at each second, the pool receives a (potentially delayed) oracle signal based on the market log mid-price $s_t = \log(\frac{1}{2}[B_t + A_t])$, where $B_t$ and $A_t$ denote the external bid and ask prices at time $t$.

For a posting delay $\delta\geq0$, define the candidate oracle signal $\widetilde{s}_t := s_{t-\delta}$, where $\delta$ is the delay between observing the market price and posting the oracle update on-chain. Following the stale-oracle model of Section~\ref{sec:lvr-stale}, oracle updates occur according to the update policy~\eqref{eq:update-policy} with heartbeat $h>0$ and deviation threshold $\varepsilon\geq0$, where the infimum is taken over the one-second observation grid. We consider the following three oracle regimes:
\begin{itemize}
\item \emph{L1 oracle}: To represent a high-latency blockchain environment, we take $h=1$ hour, $\varepsilon=25$ bps, and $\delta=60$ seconds.

\item \emph{L2 oracle}: To represent lower-latency layer-2 infrastructure, we take $h=1$ hour, $\varepsilon=5$ bps, and $\delta=2$ seconds.

\item \emph{Pull oracle}: Recent advances by, e.g., Chainlink have allowed for oracles that pull data directly from the underlying market at the time of the function call. To model this, we set $h = 1$ second (i.e., the minimal data resolution) and $\varepsilon = 0$; however, we still impose a $\delta = 1$ second delay so as to introduce some oracle imperfection.
\end{itemize}
The oracle latency $\delta$ imposed in the L1 and L2 regimes can be interpreted, in part, as arising from the delayed activation discussed within Remark~\ref{rem:MEV-mitigation}. Within each regime, we evaluate the CPMM-based GA parametrized by the oracle sensitivity $\lambda \in [0,1)$ and a proportional fee $\gamma \in [0,1)$ as introduced in Section~\ref{sec:lvr-oev}.\footnote{As in Section~\ref{sec:lvr-oev}, the fee is assessed as a fraction of the shares of the asset being sold to the pool. Further, as in Uniswap~V3, these fees are held in escrow for the liquidity providers and are \emph{not} reinvested in the pool.} For these purposes, we create a grid of oracle sensitivities $\lambda \in \{0\%,1\%,\ldots,99\%\}$ and fees $\gamma\in\{0\,\mathrm{bps},1\,\mathrm{bps},\ldots,50\,\mathrm{bps}\}$.

In order to isolate the resilience of the GA to informed order flow, we impose a strictly adversarial trading environment in which all trading volume originates from arbitrageurs exploiting stale quotes. At each second, define the fee-free pre-trade mid-quote of the pool by
$$P_t^{\mathrm{mid}} := \exp\bigl( \lambda\pi_t+(1-\lambda)\log(r_{t-1}) \bigr),$$
where $r_{t-1}$ is the reserve ratio carried into second $t$. The fee-adjusted ask and bid prices are
$$P_t^{\mathrm{ask}} = \frac{P_t^{\mathrm{mid}}}{1-\gamma}, \qquad
P_t^{\mathrm{bid}} = (1-\gamma)P_t^{\mathrm{mid}}.$$
Arbitrageurs trade whenever $B_t>P_t^{\mathrm{ask}}$ or $P_t^{\mathrm{bid}}>A_t$. At each second, the arbitrageur executes the minimal trade required to bring the relevant GA quote into line with the corresponding NBBO boundary, thereby extracting the maximal guaranteed profit available from the pool. For simplicity, we assume that the external market is infinitely deep at the NBBO and impose no additional slippage in the external market.

This simulation is run independently across all 250 trading days in 2023. To initialize each day, we assume that both the oracle and the internal reserve-based price exactly match the mid-price at market open; all performance statistics are aggregated at market close. We reset the simulation daily so as to avoid the ``overnight oracle problem'' in which the oracle would jump severely at market open due to trading activity outside of the collected dataset. We wish to note that, as a consequence, this experiment does not measure overnight or closed-market oracle risk.

From these independent daily simulations, we aggregate three primary performance metrics to evaluate the GA designs:
\begin{itemize}
\item the \emph{tracking error} (TE), defined as the root mean-square error (RMSE) between the quoted log-price of the GA and the true NBBO log mid-price;
\item the \emph{realized loss-versus-rebalancing}, measuring the difference between the value of a continuous rebalancing strategy (executed via market orders at the NBBO bid and ask prices) and the realized pool value net of collected fees; and
\item the average \emph{relative market depth}, quantified by the ratio $L_{\mathrm{GA}}/L$ of Section~\ref{sec:frontier} between the equivalent CPMM liquidity parameter required to match the local trading-curve curvature of the GA at the true NBBO mid-price and that of the equally capitalized CPMM.
\end{itemize}
Figure~\ref{fig:contour} displays the dependence of the TE and realized LVR on the GA parameters $(\lambda,\gamma)$ for each oracle regime. As expected, the TE generally increases with fees and decreases with the oracle sensitivity; though visibly worse for the L1 oracle at high values of $\lambda$, this metric is broadly stable across the oracle regimes. In contrast, the realized LVR exhibits a distinct phase shift: the L1 oracle requires defensive parameters, i.e., high fees $\gamma$ and low oracle sensitivity $\lambda$, to mitigate the losses to arbitrageurs, while the L2 and Pull oracles demonstrate that sufficiently reliable oracles can effectively reduce these losses. Finally, the empirical market depth is almost entirely determined by the $(1-\lambda)^{-1}$ multiplier, with only small variations due to the oracle delays and imposed fees. In fact, the average relative market depth falls below this multiplier by at most 0.22\% under the L1 oracle and by at most 0.012\% under the L2 and Pull oracles.

\begin{figure}[t]
\centering
\begin{subfigure}[t]{0.3\textwidth}
\centering
\includegraphics[width=\textwidth]{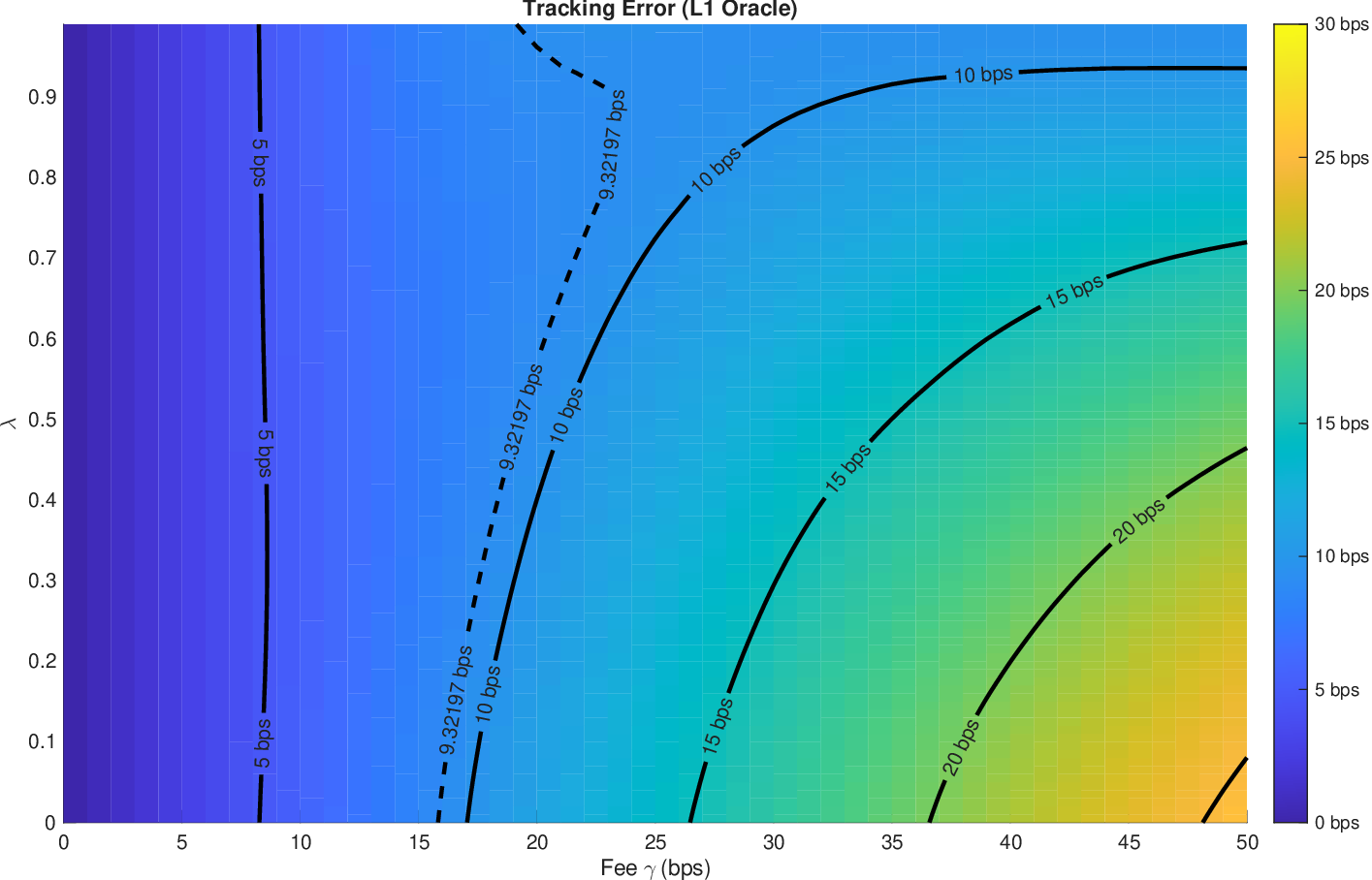}
\caption{L1 oracle, average daily TE}
\label{fig:L1-E1}
\end{subfigure}
~~~~
\begin{subfigure}[t]{0.3\textwidth}
\centering
\includegraphics[width=\textwidth]{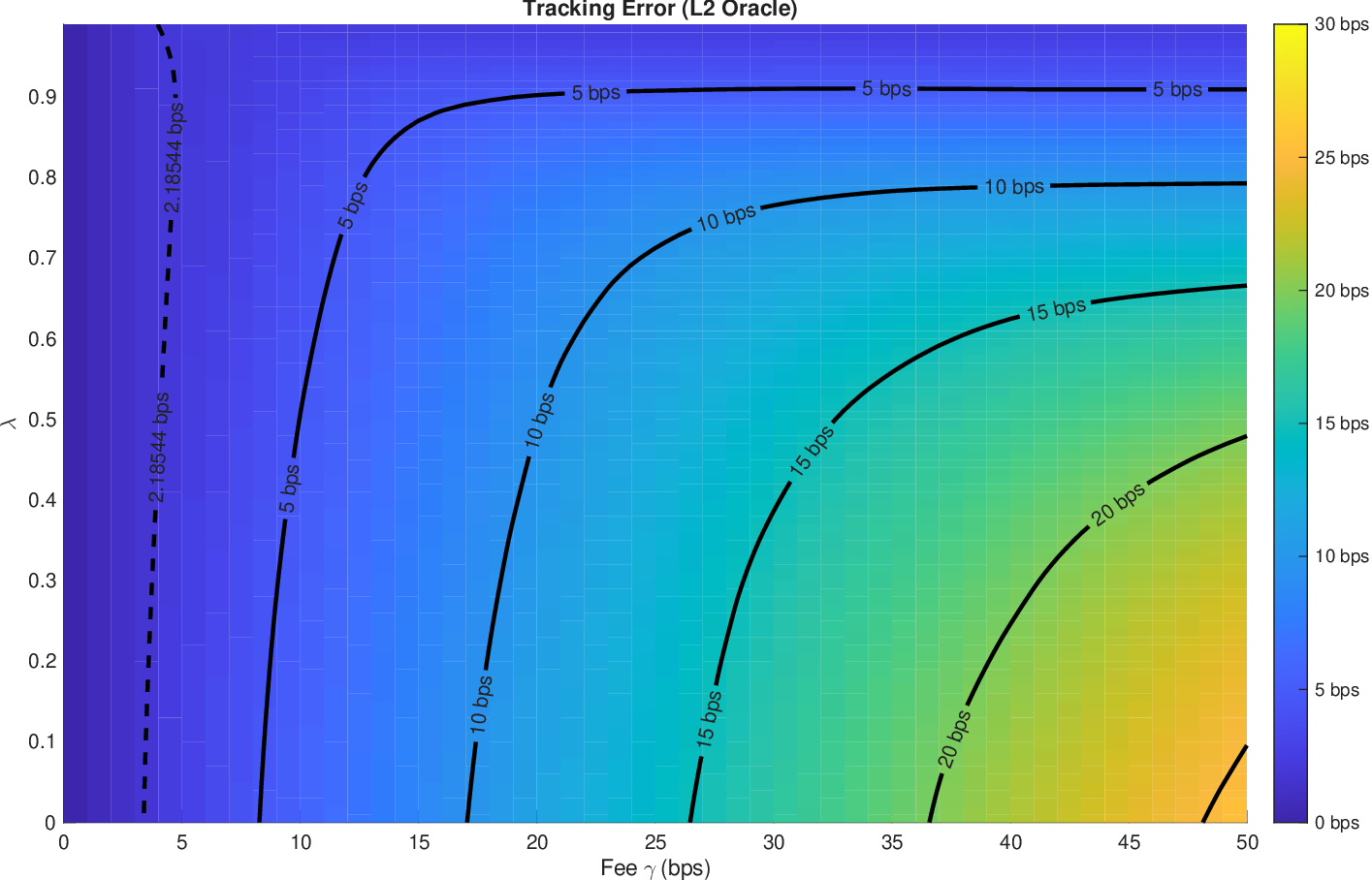}
\caption{L2 oracle, average daily TE}
\label{fig:L2-E1}
\end{subfigure}
~~~~
\begin{subfigure}[t]{0.3\textwidth}
\centering
\includegraphics[width=\textwidth]{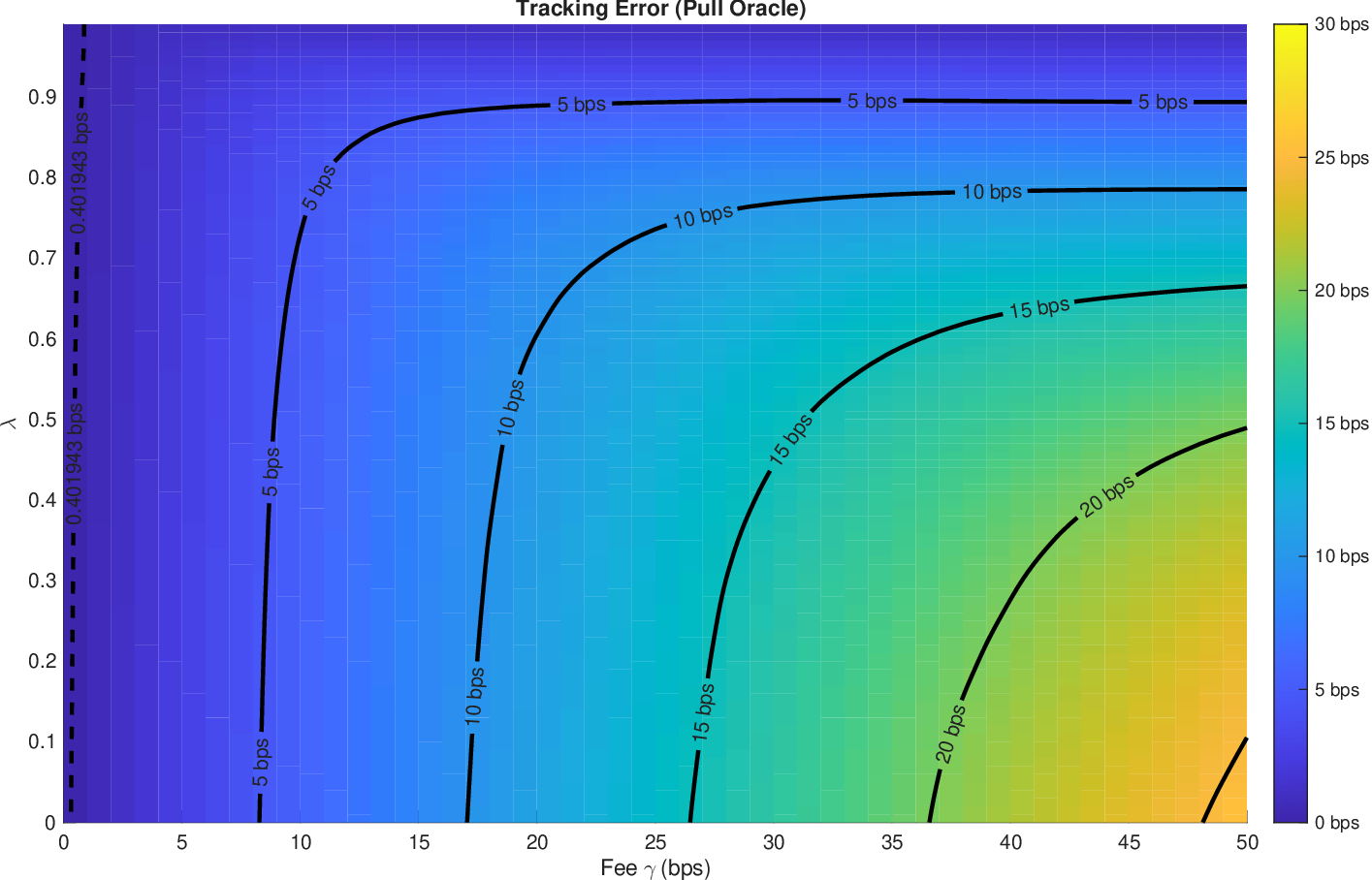}
\caption{Pull oracle, average daily TE}
\label{fig:Pull-E1}
\end{subfigure}
~~~~
\begin{subfigure}[t]{0.3\textwidth}
\centering
\includegraphics[width=\textwidth]{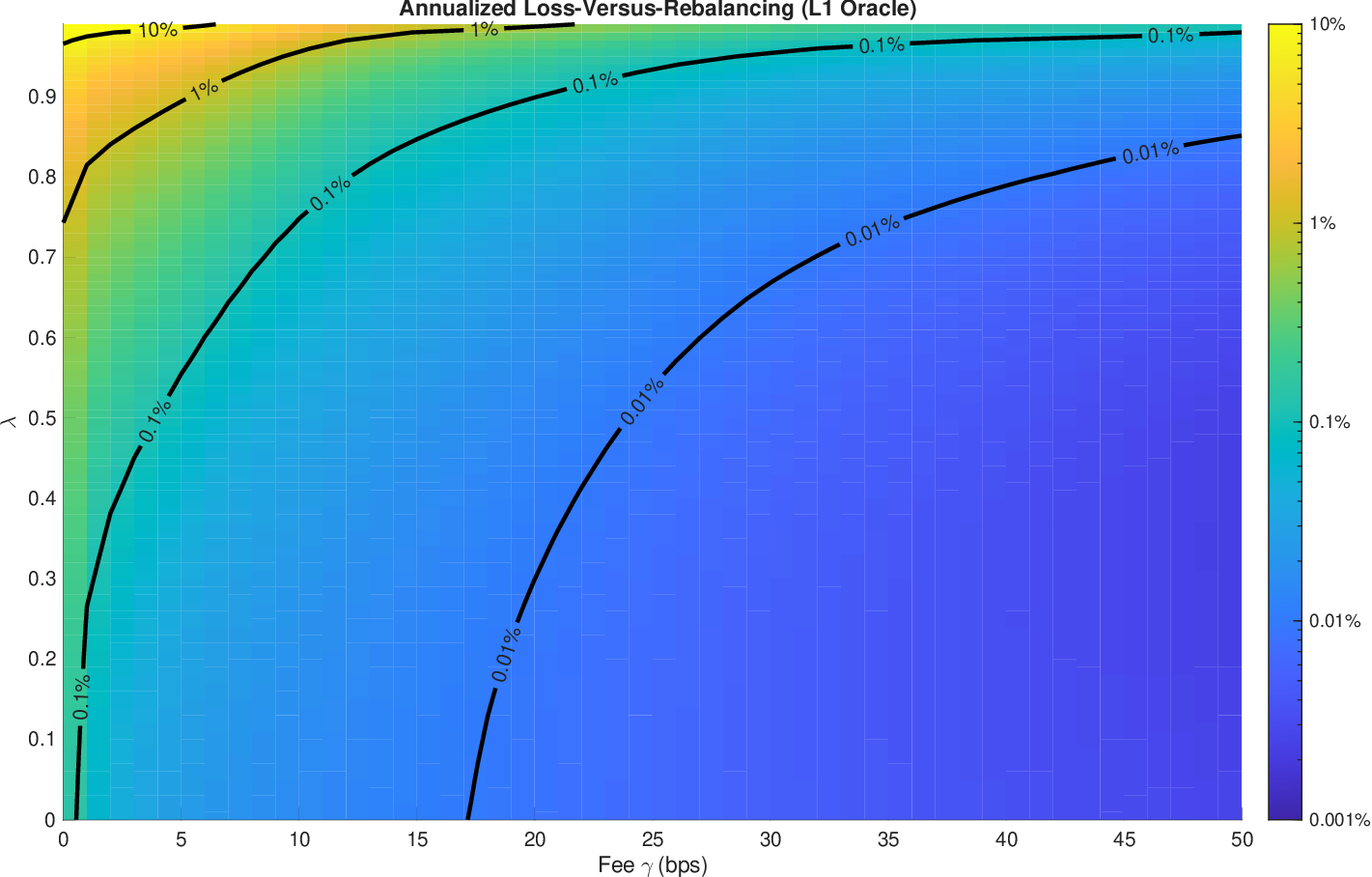}
\caption{L1 oracle, annualized LVR}
\label{fig:L1-LVR}
\end{subfigure}
~~~~
\begin{subfigure}[t]{0.3\textwidth}
\centering
\includegraphics[width=\textwidth]{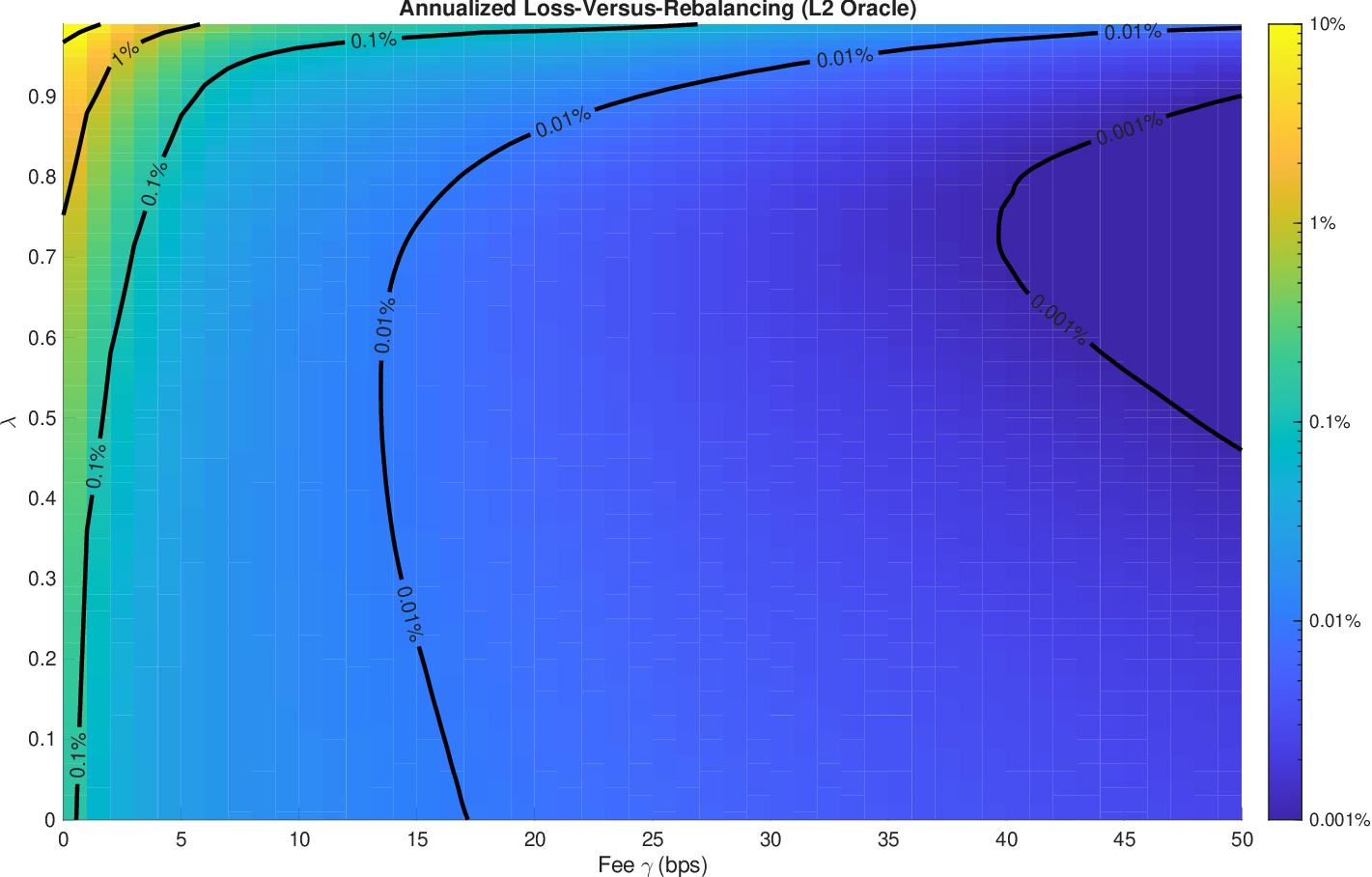}
\caption{L2 oracle, annualized LVR}
\label{fig:L2-LVR}
\end{subfigure}
~~~~
\begin{subfigure}[t]{0.3\textwidth}
\centering
\includegraphics[width=\textwidth]{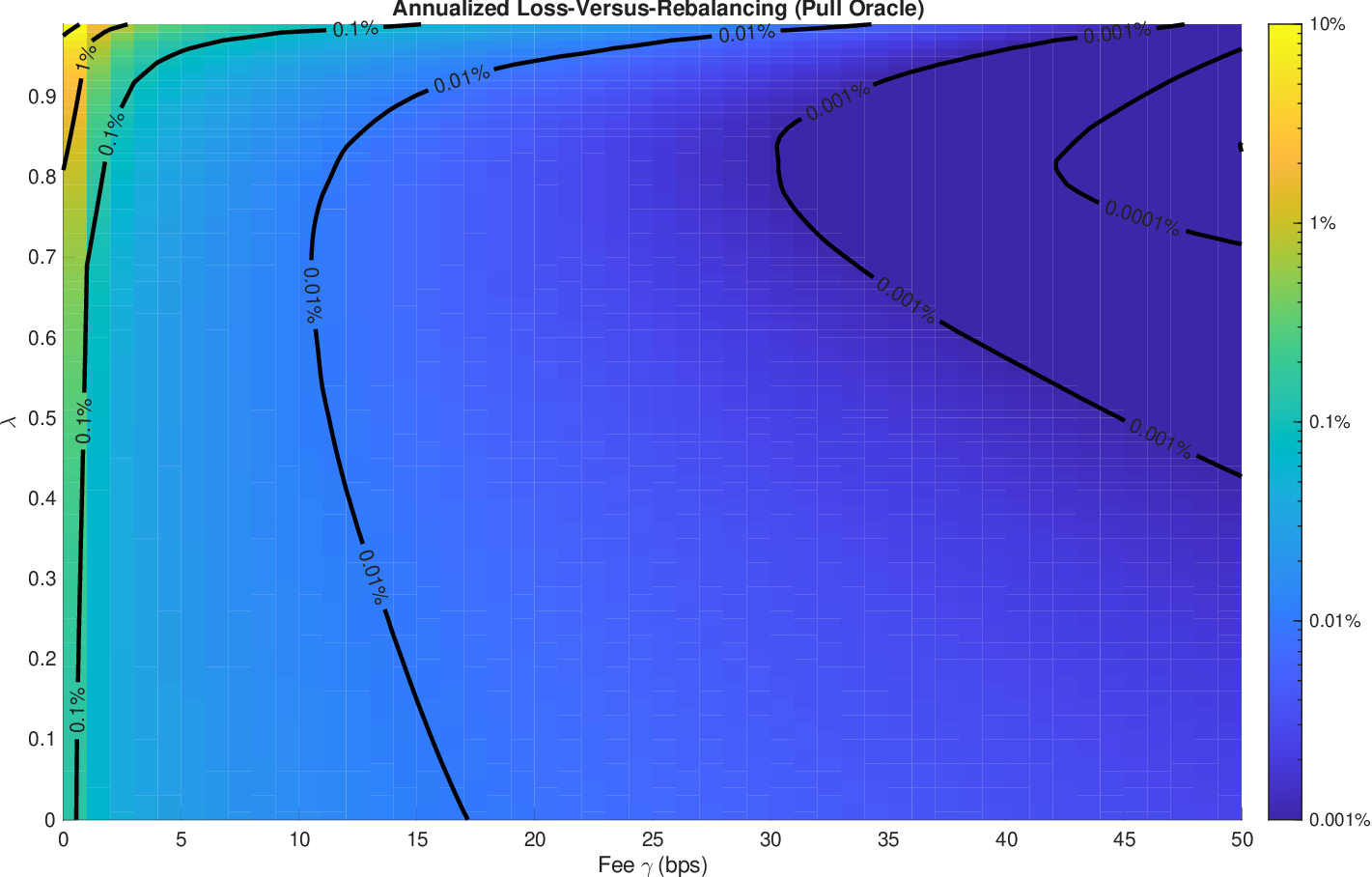}
\caption{Pull oracle, annualized LVR}
\label{fig:Pull-LVR}
\end{subfigure}
\caption{Section~\ref{sec:sp500}: Performance metrics across oracle sensitivities $\lambda\in\{0,0.01,\ldots,0.99\}$ and fees $\gamma\in\{0,1,\ldots,50\}$ bps.}
\label{fig:contour}
\end{figure}

We now focus on the empirical efficient frontiers displayed in Figure~\ref{fig:efficient-frontier} so as to translate these trends into design guidance. First, given that the TE does not exceed some threshold, Figure~\ref{fig:LVR-TE} plots the minimal average realized LVR attained across the parameter space of oracle sensitivities $\lambda$ and fees $\gamma$. Notably, by accepting a modest TE of approximately 7.5 bps, the Pull oracle regime attains negligible LVR of roughly $0.001$ bps \emph{annualized}, i.e., a liquidity provider with a \$1,000,000 position would lose just \$0.10 to arbitrage over the course of a full year. As shown in Figure~\ref{fig:LVR-TE-Optimal-Lambda}, this performance is achieved at a high oracle sensitivity $(\lambda\approx0.84)$. In contrast, the information-agnostic CPMM ($\lambda = 0$) remains optimal under the L1 oracle until the permitted TE exceeds 11 bps. At all but a few points in the L1 oracle, whenever the selected oracle sensitivity satisfies $\lambda^*>0$, the fee reaches the upper boundary of $50$ bps of the tested grid; as the simulation contains no liquidity-motivated demand, increasing the fee carries no penalty in trading volume and this boundary solution should not be read as an equilibrium fee recommendation.

Second, in Figure~\ref{fig:CE-LVR}, we invert our perspective to study the market depth that can be obtained for a given arbitrage budget. Immediately, we find that the realized market depth improves by orders of magnitude as the reliability of the oracle improves. Here, the L1 oracle can lead to higher LVR than the benchmark CPMM if improperly tuned, though it is capable of matching the risk profile of the CPMM at a significantly higher level of capital efficiency. Finally, across this entire frontier, the optimal fee level sits at $\gamma^* = 50$ bps, i.e., the upper boundary of the tested grid.

\begin{figure}[t]
\centering
\begin{subfigure}[t]{0.45\textwidth}
\centering
\includegraphics[width=\textwidth]{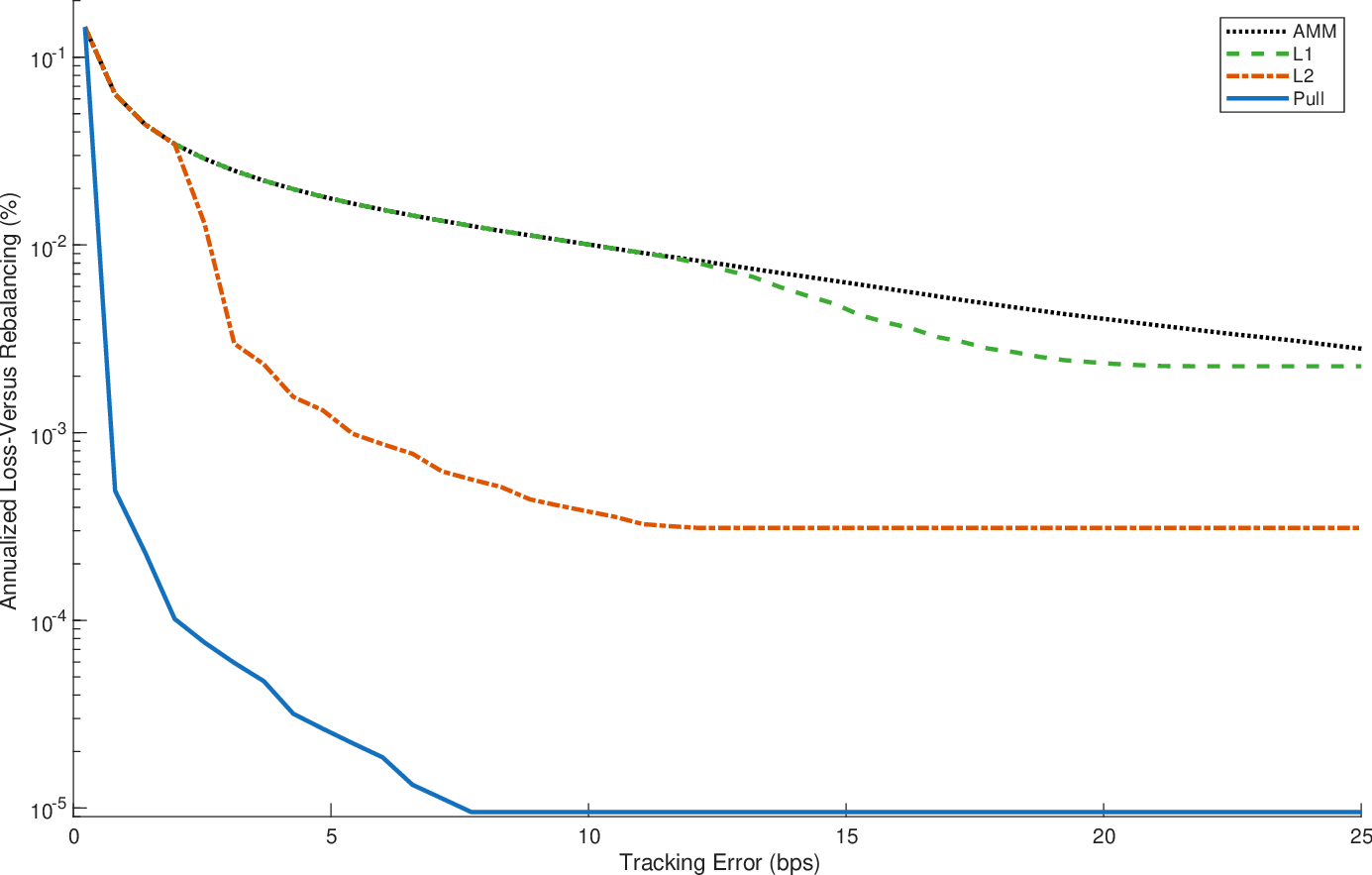}
\caption{Annualized LVR as a function of average daily TE}
\label{fig:LVR-TE}
\end{subfigure}
~~~~
\begin{subfigure}[t]{0.45\textwidth}
\centering
\includegraphics[width=\textwidth]{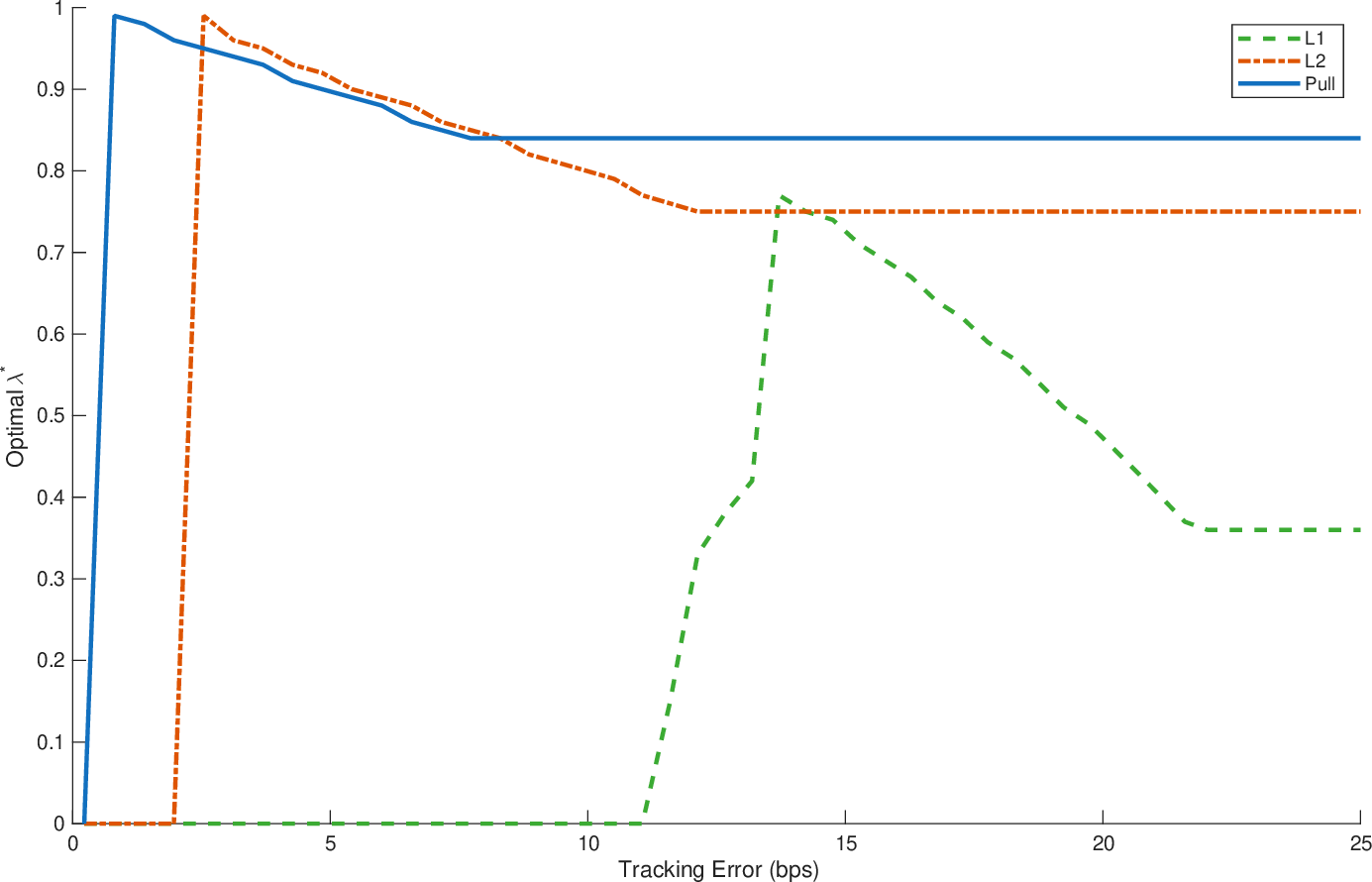}
\caption{Optimal oracle sensitivities to minimize LVR at a given TE}
\label{fig:LVR-TE-Optimal-Lambda}
\end{subfigure}
~~~~
\begin{subfigure}[t]{0.45\textwidth}
\centering
\includegraphics[width=\textwidth]{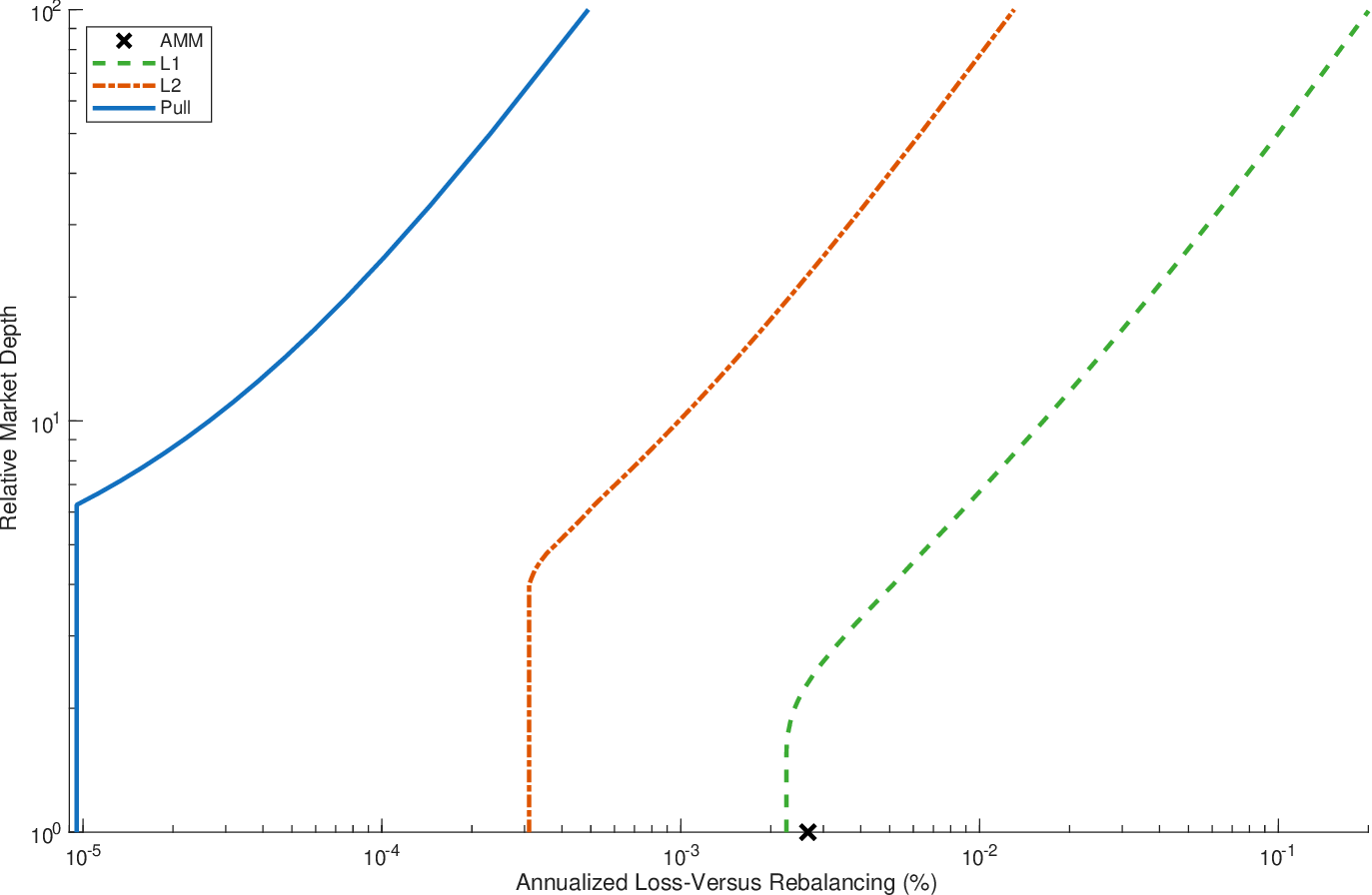}
\caption{Relative market depth as a function of the annualized LVR}
\label{fig:CE-LVR}
\end{subfigure}
~~~~
\begin{subfigure}[t]{0.45\textwidth}
\centering
\includegraphics[width=\textwidth]{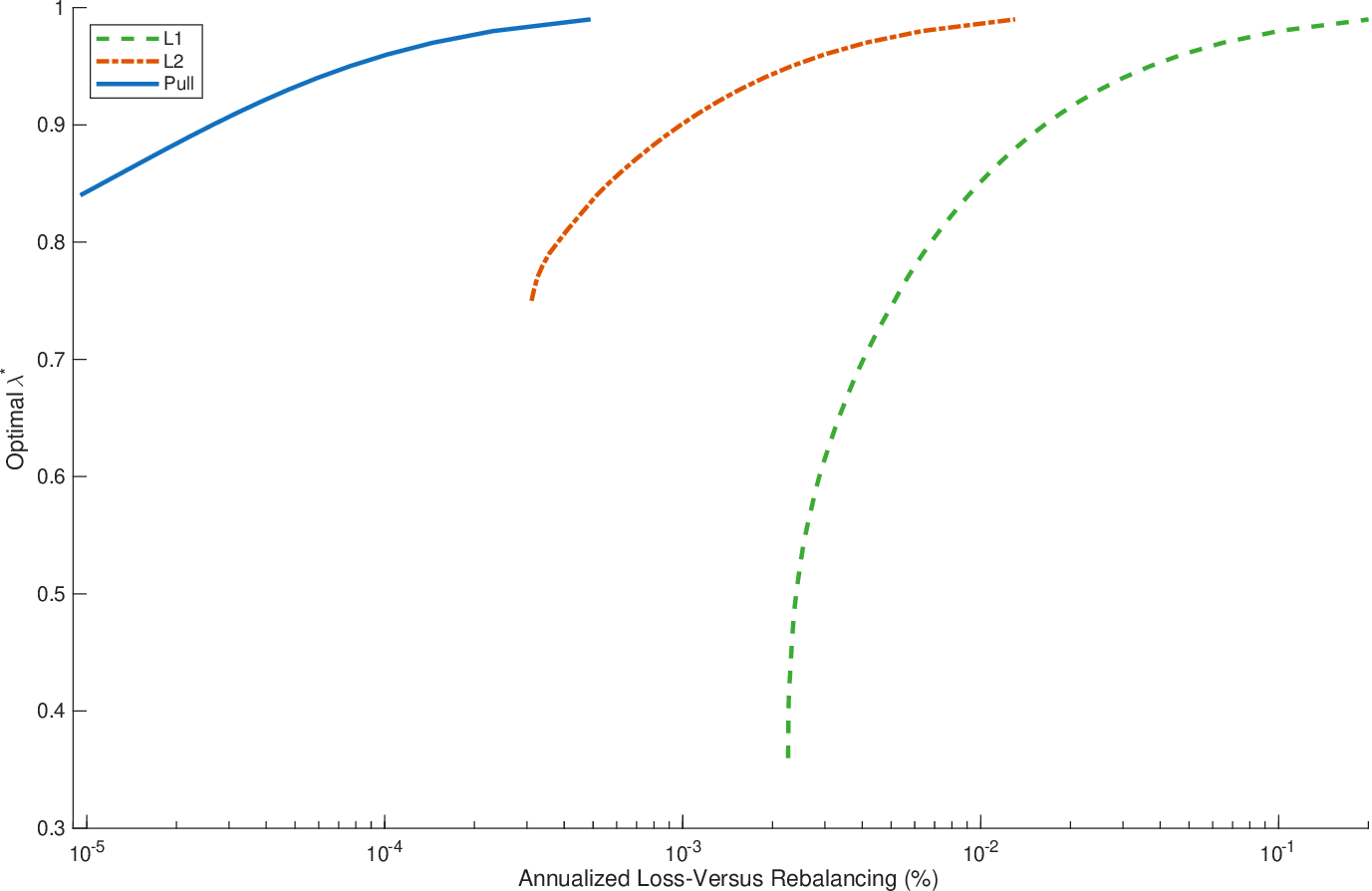}
\caption{Optimal oracle sensitivities to maximize market depth at a given LVR}
\label{fig:CE-LVR-Optimal-Lambda}
\end{subfigure}
\caption{Section~\ref{sec:sp500}: Pareto-efficient frontiers and optimal parameter configurations across oracle regimes.}
\label{fig:efficient-frontier}
\end{figure}

\section{Discussion}\label{sec:discussion}

The results of the preceding sections demonstrate that, given a reliable, low-latency oracle, the OP-AMM paradigm can create a Pareto improvement over traditional AMM designs in the two dimensions studied herein: market depth and stale-price arbitrage losses. Specifically, by simultaneously scaling market depth near the oracle price and reducing the opportunities available to latency arbitrageurs, the OP-AMM changes the execution conditions faced by traders. As a first-order approximation, this separates the market participants trading against the pool into two groups based on their trading objectives:
\begin{itemize}
\item \emph{Arbitrageurs} who simply profit from an on-chain exchange offering stale quotes are strongly disincentivized. Because the oracle tracks the external mid-price, the quote of the OP-AMM follows the external market dynamics directly. This reduces the reliance of traditional AMMs on arbitrageurs to update stale prices, the cost of which is measured by LVR. As evidenced by the substantial reduction in realized LVR under the L2 and Pull oracle regimes of Section~\ref{sec:sp500}, the OP-AMM is able to closely match the external market price at almost negligible cost (net of fees).
\item Conversely, \emph{liquidity-sensitive traders} can benefit from the deeper local liquidity of the OP-AMM near the market price. When trading against a well-calibrated pool, these participants are subject to significantly lower price slippage than when trading against an equally capitalized AMM. Importantly, this liquidity-sensitive flow is not homogeneous; it consists of both uninformed traders (whose volume provides non-toxic fee revenue to the liquidity providers) and informed traders (whose execution assists in price discovery rather than the stale-price arbitrage described above). In either case, these traders can benefit from the concentrated market depth of an OP-AMM, though their net execution quality also depends on the fee-adjusted spread.
\end{itemize}

These favorable outcomes are, however, contingent on the quality of the oracle. The LVR decomposition of Section~\ref{sec:lvr} makes precise the trade-off between the reduction of lag-induced losses and the amplification of exposure to oracle errors as the oracle sensitivity increases. As highlighted in Section~\ref{sec:cases}, the information-agnostic CPMM can have lower LVR than the GA under an insufficiently reliable oracle. Further, oracle quality is not free; the on-chain costs of maintaining a reliable, low-latency oracle are discussed within Appendix~\ref{app:gas}.

We wish to note that the above characterization splits adverse selection into two forms. The documented gains relate solely to \emph{stale-price} arbitrage, i.e., the transfer of value to arbitrageurs who merely correct outdated quotes as quantified by, e.g., the LVR of Section~\ref{sec:lvr}. Adverse selection from traders with superior information about \emph{future} prices is not addressed by our analysis and, in fact, is now met with a deeper market. Such informed flow aids in price discovery but continues to impose losses on liquidity providers in the classical market microstructure sense. Mitigating these losses would require additional predictive information or other adverse-selection controls, which are beyond the scope of this work.

Though this first-order analysis provides a plausible baseline, the OP-AMM design introduces microstructure trade-offs that warrant further study. In particular, in the arbitrage-only environment of Section~\ref{sec:sp500}, the fee selected to maximize the protection of liquidity providers subject to a TE threshold often reached the upper boundary of the tested grid, $\gamma=50$ bps. Though such a fee protects the pool from LVR, it also widens the bid-ask spread quoted by the OP-AMM and can harm the same liquidity-sensitive traders that the oracle-driven liquidity concentration is intended to support. Fully understanding the order-flow dynamics induced by the OP-AMM therefore requires an equilibrium model.

\section{Conclusion}\label{sec:conclusion}

In this work, we introduced and formalized a general framework for OP-AMMs, i.e., AMMs that augment information-agnostic designs with external price information. We showed that, given a sufficiently reliable oracle, OP-AMMs can reduce the stale-price adverse-selection costs borne by liquidity providers while increasing the local liquidity available to traders; noisy or stale oracles, however, can reverse these gains. This combination can improve the viability of decentralized exchanges for tokenized securities by mitigating the stale-price arbitrage that would otherwise affect traditional AMMs when applied to real-world assets with deep, off-chain price discovery.

Three clear extensions of this work are apparent to us. The first extension is to construct a full equilibrium model of order flow. Though our analysis indicates how OP-AMMs separate arbitrageurs from liquidity-sensitive traders, the resulting endogenous segmentation of order flow is not modeled herein. The second extension is to incorporate volatility oracles, in addition to price oracles, so as to adjust the oracle sensitivity dynamically across market regimes. Finally, though we focused on spot assets without maturity, tokenized real-world assets may also have finite maturities; market maker designs tailored to such payoff structures require separate analysis. We leave these extensions for future research.

{\footnotesize
\bibliographystyle{apalike}
\bibliography{bibtex}
}

\newpage
\appendix

\section{Proofs for Section~\ref{sec:pmm}}\label{app:proofs-pmm}

\subsection{Proof of Proposition~\ref{prop:CI}}\label{app:proof-CI}

Fix $\pi \in \R$. We first extend $U$ continuously to $\R^2_+ \times \R$; throughout, limits are taken over $(x',y') \in \R^2_{++}$. On the $y=0$ axis, for any $x>0$,
$$U(x,0,\pi) := \lim_{(x',y')\to(x,0)} U(x',y',\pi) = \log x + g_0(\pi).$$
Similarly, on the $x=0$ axis, for any $y>0$,
$$U(0,y,\pi) := \lim_{(x',y')\to(0,y)} U(x',y',\pi) = \lim_{(x',y')\to(0,y)} \bigl[\log y' - \log(y'/x') + G(y'/x',\pi)\bigr] = \log y + g_\infty(\pi).$$
Finally, at the origin $x = y = 0$, we set $U(0,0,\pi):=-\infty$, which is consistent with the positive homogeneity of $\lcal(\cdot,\cdot,\pi)$ as proven below.

We next verify that $\rcal(\pi)$ is a reachable set:
\begin{itemize}
\item \emph{Nonempty}: $(1,1)\in\rcal(\pi)$ since $U(1,1,\pi)=0$.
\item \emph{Nondegenerate}: $(0,0)\notin\rcal(\pi)$ since $U(0,0,\pi)=-\infty$.
\item \emph{Closed}: Since $U(\cdot,\cdot,\pi)$ is continuous, its superlevel set $\rcal(\pi)$ is closed.
\item \emph{Convex}: It is enough to show that $U(\cdot,\cdot,\pi)$ is concave on $\R^2_{++}$, since concavity then extends to $\R^2_+$ by continuity. Let $H(r):=\exp(G(r,\pi))$, for any $r\in\R_{++}$. Differentiation gives
$$H''(r) = -\frac{(\partial_r p(r,\pi))\exp(p(r,\pi))}{\bigl(r+\exp(p(r,\pi))\bigr)^2} H(r) \leq 0,$$
so $H$ is concave. Since $\lcal(x,y,\pi)=xH(y/x)$ for any $(x,y) \in \R^2_{++}$, $\lcal(\cdot,\cdot,\pi)$ is the perspective of $H$ and is therefore concave. Moreover, since $\lcal(\cdot,\cdot,\pi)>0$, concavity of $\lcal$ implies log-concavity, and hence $U=\log\lcal$ is concave.
\item \emph{Upward closed}: It is enough to show that $U(\cdot,\cdot,\pi)$ is non-decreasing on $\R^2_{++}$, since this property extends to $\R^2_+$ by continuity. For $(x,y)\in\R^2_{++}$,
\begin{align*}
\partial_x U(x,y,\pi) &= \frac{1}{x} \left(1 - \frac{y/x}{y/x + \exp(p(y/x,\pi))}\right) = \frac{1}{x} \frac{\exp(p(y/x,\pi))}{y/x + \exp(p(y/x,\pi))} > 0,\\
\partial_y U(x,y,\pi) &= \frac{1}{x} \frac{1}{y/x + \exp(p(y/x,\pi))} > 0.
\end{align*}
\end{itemize}

It remains to identify the canonical trading function. The function $\lcal(\cdot,\cdot,\pi)$ is continuous, concave, and non-decreasing. It is also positively homogeneous since, for any $(x,y)\in\R^2_{++}$ and $\alpha>0$,
$$\lcal(\alpha x,\alpha y,\pi) = \alpha x \exp(G(y/x,\pi)) = \alpha\lcal(x,y,\pi).$$
By continuity, this property extends to $\R^2_+$. Since $\rcal(\pi)=\bigl\{(x,y)\in\R^2_+ \mid \lcal(x,y,\pi)\geq 1\bigr\}$, it follows from~\cite[Section 1.3.3]{angeris2023geometry} that $\lcal(\cdot,\cdot,\pi)$ is the unique canonical trading function of $\rcal(\pi)$.

Finally, using the derivatives above, $\frac{\partial_x U(x,y,\pi)}{\partial_y U(x,y,\pi)}=\exp(p(y/x,\pi))$, which gives the stated marginal price.

\subsection{Proof of Theorem~\ref{thm:interp}}\label{app:proof-interp}

Fix $r>0$. For $\pi\neq\pi^*(r)$, using $p(r,\pi^*(r))=\pi^*(r)$ and the fundamental theorem of calculus,
$$\lambda(r,\pi) = \frac{p(r,\pi)-\pi^*(r)}{\pi-\pi^*(r)} = \frac{1}{\pi-\pi^*(r)} \int_{\pi^*(r)}^\pi p_\pi(r,\theta)\,d\theta = \int_0^1 p_\pi\Bigl( r,\pi^*(r)+q\bigl(\pi-\pi^*(r)\bigr) \Bigr)\,dq.$$
Assumption~\ref{ass:contraction} then gives $0\leq\lambda(r,\pi)\leq c_r<1$. Continuity away from $\pi=\pi^*(r)$ follows directly from the definition. As $\pi\to\pi^*(r)$, the integral representation converges to $p_\pi(r,\pi^*(r))$, which is precisely the value assigned to $\lambda(r,\pi^*(r))$. Hence $\pi\mapsto\lambda(r,\pi)$ is continuous.

Finally, for $\pi\neq\pi^*(r)$, rearranging the definition of $\lambda(r,\pi)$ gives $p(r,\pi)=\lambda(r,\pi)\pi+\bigl(1-\lambda(r,\pi)\bigr)\pi^*(r)$. At $\pi=\pi^*(r)$, the same identity follows from $p(r,\pi^*(r))=\pi^*(r)$, and the result follows.

\subsection{Proof of Lemma~\ref{lemma:liquidity}}\label{app:proof-liquidity}

Since $r=y/x$ and $Sx+y=1$, it follows that $x=\frac{1}{S+r}$ and $y=\frac{r}{S+r}$. At a state satisfying $p(r,\pi)=\log S$, the conditional invariant gives
$$\frac{dx}{dr} = -\frac{x}{S+r}, \qquad \frac{dy}{dx} = -e^{p(r,\pi)} = -S.$$
Hence, by the chain rule,
$$\frac{d^2y}{dx^2} = \frac{\frac{d}{dr}\left(\frac{dy}{dx}\right)}{\frac{dx}{dr}} = \frac{-Sp_r(r,\pi)}{-x/(S+r)} = \frac{Sp_r(r,\pi)(S+r)}{x}.$$
Using $x=1/(S+r)$, we recover $\kappa^{\mathrm{OP}}=\frac{p_r(r,\pi)S(S+r)^2}{(1+S^2)^{3/2}}$.

A CPMM with liquidity parameter $\widetilde L$ has curvature $\kappa^{\mathrm{CPMM}} = \frac{2S^{3/2}}{\widetilde L(1+S^2)^{3/2}}$. Matching the two curvatures and setting $\widetilde L=L^{\mathrm{OP}}$ gives $L^{\mathrm{OP}} = \frac{2\sqrt{S}}{p_r(r,\pi)(S+r)^2}$.

At price $S$, an equally capitalized CPMM with liquidity parameter $L$ has reserves $x=\frac{L}{\sqrt{S}}$ and $y=L\sqrt{S}$. The normalization $Sx+y=1$ implies $1 = S\frac{L}{\sqrt{S}} + L\sqrt{S} = 2L\sqrt{S}$, so $L=\frac{1}{2\sqrt{S}}$. Therefore, $L^{\mathrm{OP}} = \frac{4S}{p_r(r,\pi)(S+r)^2}\,L$. Thus, $L^{\mathrm{OP}}>L$ if and only if $p_r(r,\pi) < \frac{4S}{(S+r)^2}$, and the result follows.

\section{Proofs for Section~\ref{sec:lvr}}\label{app:proofs-lvr}

\subsection{Proof of Theorem~\ref{thm:lvr}}\label{app:lvr-proof}

We prove the result in three steps. First, we derive the stochastic dynamics of the equilibrium reserve ratio from the clearing condition $p(r_t,\pi_t)=s_t$. Second, we use preservation of the conditional invariant to obtain the dynamics of the risky reserve $x_t$. Finally, we combine these dynamics to derive the decomposition of the pool value and identify the instantaneous LVR rate. Throughout the proof, we adopt an oracle-first, clearing-second convention, i.e., market and oracle prices update exogenously at fixed reserves, after which arbitrageurs move the pool along the conditional trading curve associated with the new oracle input until clearing is restored.

We first consider the dynamics of the reserve ratio. As $p_r>0$, the implicit function theorem implies that $r_t=r^*(s_t,\pi_t)$ is a continuous It\^o process. Throughout this step, all partial derivatives of $p$ are evaluated at $(r_t,\pi_t)$. Applying It\^o's formula to the clearing condition $p(r_t,\pi_t)=s_t$ gives
$$ds_t = p_r\,dr_t + p_\pi\,d\pi_t + \frac{1}{2}p_{rr}\,d\langle r\rangle_t + p_{r\pi}\,d\langle r,\pi\rangle_t + \frac{1}{2}p_{\pi\pi}\,d\langle\pi\rangle_t.$$
Hence
$$dr_t = \frac{1}{p_r} \left[ ds_t - p_\pi\,d\pi_t - \frac{1}{2}p_{rr}\,d\langle r\rangle_t - p_{r\pi}\,d\langle r,\pi\rangle_t - \frac{1}{2}p_{\pi\pi}\,d\langle\pi\rangle_t \right].$$
Using
\begin{align*}
ds_t = -\frac{1}{2}\sigma_t^2\,dt + \sigma_t\,dW_t, \qquad
d\pi_t = \mu_t^\pi\,dt + \sigma_t^\pi\,dW_t^\pi, \qquad
d\langle W,W^\pi\rangle_t = \rho_t\,dt,
\end{align*}
we may write
$$dr_t = a_t\,dt + b_t\,dW_t + c_t\,dW_t^\pi,$$
where $b_t = \frac{\sigma_t}{p_r}$ and $c_t = -\frac{p_\pi}{p_r}\sigma_t^\pi$.

Define
$$Q_t := \sigma_t^2 - 2\rho_t\,p_\pi\sigma_t\sigma_t^\pi + p_\pi^2(\sigma_t^\pi)^2
\qquad \text{ and } \qquad
M_t := \rho_t\,\sigma_t\sigma_t^\pi - p_\pi(\sigma_t^\pi)^2.$$
Then $d\langle r\rangle_t = \frac{Q_t}{p_r^2}\,dt$ and $d\langle r,\pi\rangle_t = \frac{M_t}{p_r}\,dt$. Equivalently, setting
$$q_t := b_t^2+c_t^2+2\rho_t b_tc_t = \frac{Q_t}{p_r^2}, \qquad m_t := \sigma_t^\pi(\rho_t b_t+c_t) = \frac{M_t}{p_r},$$
we have $d\langle r\rangle_t=q_t\,dt$ and $d\langle r,\pi\rangle_t=m_t\,dt$. Substituting into the finite-variation part of $dr_t$ gives
$$a_t = \frac{1}{p_r} \left[ -\frac{1}{2}\sigma_t^2 - p_\pi\mu_t^\pi - \frac{1}{2}p_{rr}q_t - p_{r\pi}m_t - \frac{1}{2}p_{\pi\pi}(\sigma_t^\pi)^2 \right]$$
or, equivalently,
\begin{align*}
a_t = \frac{1}{p_r} \Biggl[ &-\frac{1}{2}\sigma_t^2 - p_\pi\mu_t^\pi - \frac{1}{2}\frac{p_{rr}}{p_r^2} \bigl( \sigma_t^2 - 2\rho_t\,p_\pi\sigma_t\sigma_t^\pi + p_\pi^2(\sigma_t^\pi)^2 \bigr) \\
&- \frac{p_{r\pi}}{p_r} \bigl( \rho_t\,\sigma_t\sigma_t^\pi - p_\pi(\sigma_t^\pi)^2 \bigr) - \frac{1}{2}p_{\pi\pi}(\sigma_t^\pi)^2 \Biggr].
\end{align*}

We next use preservation of the conditional invariant to obtain the dynamics of the risky reserve. Write
$$U(x,y,\pi) = \log x + G\left(\frac{y}{x},\pi\right), \qquad G_r(r,\pi) = \frac{1}{r+e^{p(r,\pi)}}.$$
Under the oracle-first, clearing-second convention, an update from $t$ to $t+dt$ satisfies
$$U(x_{t+dt},r_{t+dt}x_{t+dt},\pi_{t+dt}) = U(x_t,r_tx_t,\pi_{t+dt}).$$
The right-hand side evaluates the updated conditional invariant at the pre-trade reserves, while the left-hand side evaluates it at the post-clearing reserves. Thus, the oracle update changes the family of invariants but does not itself constitute a trade.

Taking the It\^o limit gives
$$d\log x_t + G_r(r_t,\pi_t)\,dr_t + \frac{1}{2}G_{rr}(r_t,\pi_t)\,d\langle r\rangle_t + G_{r\pi}(r_t,\pi_t)\,d\langle r,\pi\rangle_t = 0.$$
The direct terms $G_\pi\,d\pi_t$ and $\frac{1}{2}G_{\pi\pi}\,d\langle\pi\rangle_t$ appear on both sides and cancel. Using the equilibrium condition $p(r_t,\pi_t)=s_t$, we obtain $G_r(r_t,\pi_t) = \frac{1}{r_t+e^{s_t}}$, $G_{rr}(r_t,\pi_t) = - \frac{1+p_r(r_t,\pi_t)e^{s_t}}{(r_t+e^{s_t})^2}$, and $G_{r\pi}(r_t,\pi_t) = - \frac{p_\pi(r_t,\pi_t)e^{s_t}}{(r_t+e^{s_t})^2}$. Therefore,
$$d\log x_t = - \frac{1}{r_t+e^{s_t}}\,dr_t + \frac{1}{2} \frac{1+p_r(r_t,\pi_t)e^{s_t}}{(r_t+e^{s_t})^2} \,d\langle r\rangle_t + \frac{p_\pi(r_t,\pi_t)e^{s_t}}{(r_t+e^{s_t})^2} \,d\langle r,\pi\rangle_t.$$
Substituting $dr_t = a_t\,dt+b_t\,dW_t+c_t\,dW_t^\pi$, together with $d\langle r\rangle_t=q_t\,dt$ and $d\langle r,\pi\rangle_t=m_t\,dt$, gives
$$d\log x_t = \alpha_t\,dt - \frac{b_t}{r_t+e^{s_t}}\,dW_t - \frac{c_t}{r_t+e^{s_t}}\,dW_t^\pi,$$
where
$$\alpha_t = -\frac{a_t}{r_t+e^{s_t}} + \frac{1}{2} \frac{1+p_r(r_t,\pi_t)e^{s_t}}{(r_t+e^{s_t})^2} q_t + \frac{p_\pi(r_t,\pi_t)e^{s_t}}{(r_t+e^{s_t})^2} m_t.$$

Finally, recall that the pool value is $V_t = x_tS_t+y_t = x_t(S_t+r_t)$. Applying It\^o's formula gives $dV_t = (S_t+r_t)\,dx_t + x_t\,dS_t + x_t\,dr_t + d\langle x,S\rangle_t + d\langle x,r\rangle_t$. Since $dx_t = x_t\,d\log x_t + \frac{1}{2}x_t\,d\langle\log x\rangle_t$ and $d\langle\log x\rangle_t = \frac{q_t}{(S_t+r_t)^2}\,dt$, the local martingale part is
$$dV_t^{\mathrm{mart}} = x_tS_t\sigma_t\,dW_t + (S_t+r_t)x_t \left( -\frac{b_t}{S_t+r_t}\,dW_t - \frac{c_t}{S_t+r_t}\,dW_t^\pi \right) + x_t(b_t\,dW_t+c_t\,dW_t^\pi).$$
The terms involving $b_t$ and $c_t$ cancel, so $dV_t^{\mathrm{mart}} = x_t\,dS_t$.

For the finite-variation part,
$$d\langle x,S\rangle_t = - \frac{x_tS_t\sigma_t(b_t+\rho_t c_t)}{S_t+r_t}\,dt, \qquad d\langle x,r\rangle_t = - \frac{x_tq_t}{S_t+r_t}\,dt.$$
Consequently,
$$\frac{1}{x_t} \frac{dV_t^{\mathrm{drift}}}{dt} = (S_t+r_t) \left( \alpha_t + \frac{q_t}{2(S_t+r_t)^2} \right) + a_t - \frac{S_t\sigma_t(b_t+\rho_t c_t)}{S_t+r_t} - \frac{q_t}{S_t+r_t}.$$
Substituting the expression for $\alpha_t$ and simplifying gives
$$\frac{1}{x_t} \frac{dV_t^{\mathrm{drift}}}{dt} = \frac{S_t}{S_t+r_t} \left[ \frac{1}{2}p_rq_t + p_\pi m_t - \sigma_t(b_t+\rho_t c_t) \right].$$
Using
\begin{align*}
b_t = \frac{\sigma_t}{p_r}, \qquad
c_t = -\frac{p_\pi}{p_r}\sigma_t^\pi, \qquad
q_t = \frac{Q_t}{p_r^2}, \qquad
m_t = \frac{M_t}{p_r},
\end{align*}
we obtain $\frac{1}{2}p_rq_t + p_\pi m_t - \sigma_t(b_t+\rho_t c_t) = -\frac{Q_t}{2p_r}$. Therefore, $dV_t^{\mathrm{drift}} = - \frac{1}{2} \frac{x_tS_t}{S_t+r_t} \frac{Q_t}{p_r(r_t,\pi_t)} \,dt$.

By the definition of the residual innovation process $\zeta_t$ in Section~\ref{sec:lvr-general}, $Q_t = \frac{d}{dt}\langle\zeta\rangle_t$. Hence $dV_t = x_t\,dS_t - \ell_t\,dt$, where
$$\ell_t = \frac{1}{2} \frac{x_tS_t}{S_t+r_t} \frac{1}{p_r(r_t,\pi_t)} \frac{d}{dt}\langle\zeta\rangle_t.$$
Equivalently,
$$\ell_t = \frac{1}{2} \frac{x_tS_t}{S_t+r_t} \frac{1}{p_r(r_t,\pi_t)} \bigl( \sigma_t^2 - 2\rho_t\,p_\pi(r_t,\pi_t)\sigma_t\sigma_t^\pi + p_\pi(r_t,\pi_t)^2(\sigma_t^\pi)^2 \bigr),$$
which proves the result.

\subsection{Proof of Proposition~\ref{prop:stale-lvr}}\label{app:proof-stale-lvr}

We prove the result by considering separately the continuous and jump components of LVR. For $t\in(\tau_k,\tau_{k+1})$, the oracle is constant, so $\mu_t^\pi=\sigma_t^\pi=0$. Hence $d\pi_t=0$, and the residual innovation process of Section~\ref{sec:lvr-general} satisfies $d\zeta_t=ds_t$. Therefore, $\frac{d}{dt}\langle\zeta\rangle_t=\sigma_t^2$, and the continuous LVR rate follows directly from Theorem~\ref{thm:lvr}.

Next, consider an update time $\tau_k$. During the instantaneous clearing adjustment, the external price remains fixed at $S_{\tau_k}$, and the reserves move along the post-update conditional trading curve from $r_{\tau_k^-}$ to $r_{\tau_k}$. Along this curve, the pool value is $V(r)=x_k(r)(S_{\tau_k}+r)$. Since $x_k(r) = -\Bigl(r+e^{p(r,\pi_{\tau_k})}\Bigr)x_k'(r)$, we have
$$V'(r) = x_k'(r)(S_{\tau_k}+r)+x_k(r) = \Bigl(S_{\tau_k}-e^{p(r,\pi_{\tau_k})}\Bigr)x_k'(r).$$
Thus,
$$\Delta V_{\tau_k} = \int_{r_{\tau_k^-}}^{r_{\tau_k}} \Bigl( S_{\tau_k} - e^{p(r,\pi_{\tau_k})} \Bigr) x_k'(r)\,dr,$$
which gives~\eqref{eq:stale-jump-lvr} since $\Delta\mathrm{LVR}_{\tau_k} = -\Delta V_{\tau_k}$.

By monotonicity of $p(\cdot,\pi_{\tau_k})$, the arbitrage path moves toward the unique state satisfying $p(r_{\tau_k},\pi_{\tau_k})=s_{\tau_k}$. In particular, $e^{p(r,\pi_{\tau_k})}-S_{\tau_k}$ has the opposite sign of $r_{\tau_k}-r_{\tau_k^-}$ along the path. Since $x_k'(r)<0$, the integrand in~\eqref{eq:stale-jump-lvr} has the same sign as $r_{\tau_k}-r_{\tau_k^-}$, and therefore $\Delta\mathrm{LVR}_{\tau_k}\geq 0$.

Finally, as the external price process is continuous, the rebalancing benchmark $R_t = V_0+\int_0^t x_u\,dS_u$ does not jump at the update times. Hence the jump of $\mathrm{LVR}_t=R_t-V_t$ at $\tau_k$ is $-\Delta V_{\tau_k} = \Delta\mathrm{LVR}_{\tau_k}$. Summing the continuous and jump contributions over $[0,T]$ gives the result.

\subsection{Proof of Proposition~\ref{prop:MEV-0fee}}\label{app:proof-MEV-0fee}

We begin with the rescaling property that underlies the geometry of Section~\ref{sec:lvr-oev}. As $\lcal(\cdot,\cdot,\pi)=\exp U(\cdot,\cdot,\pi)$ is positively homogeneous, the superlevel sets of the conditional invariant are rescalings of the reachable set, i.e.,
$$\bigl\{(x,y) \in \R^2_+ \mid U(x,y,\pi) \geq c\bigr\} = e^c\,\rcal(\pi), \qquad c \in \R.$$
Positive homogeneity then implies, for any reserve state $(x,y)\in\R^2_+$,
\begin{equation}\label{eq:pvf-rescaled}
\min\bigl\{Sx'+y' \mid U(x',y',\pi) \geq U(x,y,\pi)\bigr\} = \lcal(x,y,\pi)\,\V_\pi(S).
\end{equation}

We prove the two claims separately.
\begin{enumerate}
\item Fix $(x',y')\in\lcal_-\rcal(\pi_-)$ and write $\lcal_+:=\lcal(x',y',\pi_+)$ so that
$$\Pi(x',y') = \bigl[Sx_-+y_-\bigr] - \inf\bigl\{Sx_++y_+ \mid (x_+,y_+) \in \lcal_+\rcal(\pi_+)\bigr\}.$$
By~\eqref{eq:pvf-rescaled} applied at $(x',y')$ and $\pi_+$, this infimum equals $\lcal_+\V_{\pi_+}(S)$ and is attained at $\lcal_+$ times the minimizer over $\rcal(\pi_+)$, i.e., at the unique point of reserve ratio $r^*(s,\pi_+)=r_+$.
\item Suppose $\dpi>0$. By part~\eqref{prop:MEV-backrun},
$$\Pi^* = \bigl[Sx_-+y_-\bigr] - \V_{\pi_+}(S) \inf_{(x',y')\in\lcal_-\rcal(\pi_-)} \lcal(x',y',\pi_+).$$
It is enough to consider this infimum over the efficient boundary. Indeed, any $(x,y)\in\lcal_-\rcal(\pi_-)$ with $\lcal(x,y,\pi_-)>\lcal_-$ may be rescaled by $\alpha := \frac{\lcal_-}{\lcal(x,y,\pi_-)} <1$, which keeps it feasible and, by positive homogeneity, replaces $\lcal(x,y,\pi_+)$ by the strictly smaller $\alpha\lcal(x,y,\pi_+)$.

Parametrize the efficient boundary within $\R^2_{++}$ by the reserve ratio. Let $\xi(r)$ solve $U\bigl(\xi(r),r\xi(r),\pi_-\bigr) = \log\lcal_-$. Then $\log\xi(r) = \log \lcal_- - G(r,\pi_-)$, and
$$\log\lcal\bigl(\xi(r),r\xi(r),\pi_+\bigr) = \log\lcal_- + G(r,\pi_+) - G(r,\pi_-).$$
Differentiating gives
$$\frac{d}{dr}\log \lcal\bigl(\xi(r),r\xi(r),\pi_+\bigr) = \frac{1}{r+\exp(p(r,\pi_+))} - \frac{1}{r+\exp(p(r,\pi_-))}.$$
Since $\dpi>0$, monotonicity of $p$ in its oracle argument implies $p(r,\pi_+)\geq p(r,\pi_-)$, so this derivative is non-positive. Hence the infimum over the efficient boundary is approached as $r\nearrow\infty$ with value $\lcal_-e^D$. Suppose first that $g_\infty(\pi_-)>-\infty$. Then $G(r,\pi_-)-\log r$ converges, so that $\xi(r)\to0$ and
$$r\xi(r) = \exp\bigl(\log\lcal_-+\log r-G(r,\pi_-)\bigr) \longrightarrow \lcal_-e^{-g_\infty(\pi_-)}.$$
Hence $\lcal_-\rcal(\pi_-)$ contains the boundary point $\bigl(0,\lcal_-e^{-g_\infty(\pi_-)}\bigr)$. At this point,
$$\lcal(0,y,\pi_+) = ye^{g_\infty(\pi_+)} = \lcal_-e^{g_\infty(\pi_+)-g_\infty(\pi_-)},$$
so the infimum is attained and $D = g_\infty(\pi_+)-g_\infty(\pi_-)$. If instead $g_\infty(\pi_-)=-\infty$, then $\lcal_-\rcal(\pi_-)$ does not meet the $y$-axis, while the same limiting argument gives the value $\lcal_-e^D$ as $r\nearrow\infty$. If, in addition, $p(r,\pi_+)>p(r,\pi_-)$ for every $r>0$, then the derivative above is strictly negative for every finite $r$, and the infimum is not attained at any finite reserve state. Finally, the condition for $\Pi^*$ to equal the entire pool value follows directly from its displayed expression since $\V_{\pi_+}(S)>0$ and $\lcal_->0$.
\end{enumerate}

The case $\dpi<0$ follows symmetrically.

\subsection{Proof of Lemma~\ref{lemma:MEV+fees}}\label{app:proof-MEV-fees}

Fix $(x',y')\in\ccal_-$ and set $r' := \frac{y'}{x'} \geq r_-$. Let $C_+ := U(x',y',\pi_+)$ and, for notational convenience, define $P_-(r) := \exp(p(r,\pi_-))$ and $P_+(r) := \exp(p(r,\pi_+))$.

We first determine the optimal back-run. As the back-run sells the risky asset to the pool, the fee is paid in the risky asset and the back-run minimizes $\frac{S}{1-\gamma}x_++y_+ = Se^{\Gamma}x_++y_+$ over the refreshed trading set subject to $x_+\geq x'$. Without this constraint, part~\eqref{prop:MEV-backrun} of Proposition~\ref{prop:MEV-0fee} applies at the fee-adjusted price $Se^{\Gamma}$ and gives the minimum value $\lcal(x',y',\pi_+)\V_{\pi_+}\bigl(Se^{\Gamma}\bigr)$, attained at reserve ratio $r_\gamma := r^*(s+\Gamma,\pi_+)$. Parametrizing the refreshed conditional curve by
$$x_+(q) = \exp\bigl(C_+-G(q,\pi_+)\bigr), \qquad
y_+(q) = q x_+(q),$$
the constraint $x_+\geq x'$ is equivalent to $q\leq r'$. This constraint binds precisely when $r'\leq r_\gamma$, i.e., when $p(r',\pi_+)\leq s+\Gamma$. Therefore:
\begin{itemize}
\item If $p(r',\pi_+) \leq s+\Gamma$, then the minimizer is $q=r'$ and the optimal back-run is the null trade.
\item If $p(r',\pi_+) > s+\Gamma$, then the optimal post-update reserve ratio is $r_\gamma < r'$, and the corresponding risky reserve $\chi_\gamma(x',y')$ is the unique solution of $U\bigl( \chi_\gamma, r_\gamma\chi_\gamma, \pi_+ \bigr) = C_+$.
\end{itemize}

We next optimize the front-run. Let $C_- := U(x_-,y_-,\pi_-)=\log\lcal_-$. It suffices to maximize over the efficient boundary of $\ccal_-$. Indeed, suppose $U(x',y',\pi_-)>C_-$ and let $\widehat y<y'$ solve $U(x',\widehat y,\pi_-)=C_-$, which exists because $U(x',0,\pi_-)\leq U(x_-,0,\pi_-)<C_-$. For fixed $(x_+,y_+)$, the round-trip value depends on $y'$ only through $-\frac{\gamma}{1-\gamma}y'$, so lowering $y'$ to $\widehat y$ does not decrease it. Moreover, $x'$ is unchanged, hence so is the constraint $x_+\geq x'$, while $U(x',\widehat y,\pi_+)<U(x',y',\pi_+)$ enlarges the feasible set of the back-run and therefore cannot increase its minimized cost. Hence $\Pi_\gamma(x',\widehat y)\geq\Pi_\gamma(x',y')$.

Every point of the efficient boundary is parametrized by its reserve ratio $r\geq r_-$ as $\xi(r)=\exp\bigl(C_--G(r,\pi_-)\bigr)$ and $r\xi(r)$. Write $\Pi_\gamma(r):=\Pi_\gamma\bigl(\xi(r),r\xi(r)\bigr)$.

First suppose that $p(r,\pi_+) \leq s+\Gamma$. The optimal back-run is the null trade and
$$\Pi_\gamma(r) = S\bigl(x_--\xi(r)\bigr) + \frac{y_--r\xi(r)}{1-\gamma}.$$
Differentiating gives $\Pi_\gamma'(r) = \left[ S-\frac{P_-(r)}{1-\gamma} \right] \frac{\xi(r)}{r+P_-(r)}$. As $r\geq r_-$ and $p(\cdot,\pi_-)$ is increasing, $p(r,\pi_-) \geq p(r_-,\pi_-) \geq s-\Gamma$, so that $P_-(r)\geq S(1-\gamma)$. Hence $\Pi_\gamma'(r)\leq0$.

Now suppose that $p(r,\pi_+) > s+\Gamma$. The optimal back-run ends at $r_\gamma=r^*(s+\Gamma,\pi_+)$, and monotonicity of $p(\cdot,\pi_+)$ implies $r>r_\gamma$. Define
$$\chi_\gamma(r):= \chi_\gamma\bigl(\xi(r),r\xi(r)\bigr) = \exp\bigl( C_- - G(r_\gamma,\pi_+) + G(r,\pi_+) - G(r,\pi_-) \bigr).$$
The attack value is
$$\Pi_\gamma(r) = Sx_- + \frac{y_-}{1-\gamma} + \frac{\gamma}{1-\gamma} (S-r)\xi(r) - \left( \frac{S}{1-\gamma} + r_\gamma \right) \chi_\gamma(r),$$
and differentiation gives
$$\Pi_\gamma'(r) = - \frac{\gamma}{1-\gamma} \frac{\xi(r)\bigl(S+P_-(r)\bigr)}{r+P_-(r)} + \left( \frac{S}{1-\gamma} + r_\gamma \right) \frac{\bigl(P_+(r)-P_-(r)\bigr)\chi_\gamma(r)}{\bigl(r+P_+(r)\bigr) \bigl(r+P_-(r)\bigr)}.$$

For fixed $r$, let $\widehat x_r(q) := \exp\bigl( U(\xi(r),r\xi(r),\pi_+) - G(q,\pi_+) \bigr)$. Then $\widehat x_r(r)=\xi(r)$ and $\widehat x_r(r_\gamma)=\chi_\gamma(r)$. Moreover,
$$\frac{d}{dq} \Bigl[ \widehat x_r(q) \bigl( q+P_+(q) \bigr) \Bigr] = \widehat x_r(q)P_+'(q) \geq 0.$$
Since $r>r_\gamma$ and $P_+(r_\gamma)=Se^{\Gamma}$, $\xi(r) \bigl( r+P_+(r) \bigr) \geq \chi_\gamma(r) \left( r_\gamma+\frac{S}{1-\gamma} \right)$, and therefore
$$\xi(r) \geq \left( r_\gamma+\frac{S}{1-\gamma} \right) \frac{\chi_\gamma(r)}{r+P_+(r)}.$$

Condition~\eqref{eq:fee-quote-shift} implies $P_+(r) \leq e^\Gamma P_-(r) = \frac{P_-(r)}{1-\gamma}$, and hence $P_+(r)-P_-(r) \leq \frac{\gamma}{1-\gamma}P_-(r)$. Combining these inequalities gives
$$\Pi_\gamma'(r) \leq - \frac{\gamma}{1-\gamma} \frac{\xi(r)\bigl(S+P_-(r)\bigr)}{r+P_-(r)} + \frac{\gamma}{1-\gamma} \frac{\xi(r)P_-(r)}{r+P_-(r)} = - \frac{\gamma}{1-\gamma} \frac{S\xi(r)}{r+P_-(r)} \leq 0.$$

At the boundary $p(r,\pi_+)=s+\Gamma$, the two expressions for $\Pi_\gamma(r)$ coincide. Thus, $\Pi_\gamma(r)$ is non-increasing for every $r\geq r_-$, and the optimal front-run is the null trade $r=r_-$. Consequently, $\Pi_\gamma^*=\Pi_\gamma(x_-,y_-)$.

Finally, if $p(r_-,\pi_+)\leq s+\Gamma$, then the optimal back-run from the no-front-run state is also the null trade. Hence $\Pi_\gamma(x_-,y_-)=0$, which implies $\Pi_\gamma^*=0$.

\section{Implementability and Gas Costs}\label{app:gas}

Though the OP-AMM framework provides a general approach to mitigating adverse selection, its practical implementation on a blockchain requires an understanding of the associated gas costs. As iterative numerical solvers can be prohibitively expensive, practical implementations benefit from specifying an OP-CFMM with a closed-form bonding curve. Notably, though we defined OP-AMMs through the log-price function, this paradigm accommodates a wide class of such implementable designs; in particular, we highlight the CPMM-based GA, which admits a CES invariant curve (see Example~\ref{ex:GAPMM}).

From the perspective of smart contract architecture, executing a swap on an OP-AMM mirrors a standard AMM but introduces infrastructure costs associated with querying the external oracle. The gas profile of this dependency is dictated by the underlying oracle architecture:
\begin{itemize}
\item \emph{Push oracles}: Traditional designs rely on data providers ``pushing'' price updates to the blockchain. Though this allows the OP-AMM to cheaply read the price from the oracle contract during a trade, it imposes a systematic maintenance cost on the ecosystem for the gas required for these updates. These frictions can introduce latency into oracle designs, e.g., as modeled by the L1 and L2 oracles of Section~\ref{sec:sp500}.
\item \emph{Pull oracles}: Modern ``pull'' oracles can enable information retrieval with very low latency. However, this architecture shifts the oracle update costs to the trader; rather than reading a stored variable on an external contract, pull oracles require the transaction itself to carry and verify a recent off-chain price report. This execution may incur higher gas costs for the trader.
\end{itemize}
Protocol designers must weigh these infrastructure costs against the surplus generated by a well-calibrated OP-AMM. We leave the study of the trade-off between gas costs and market efficiency for future research.

\end{document}